\documentclass[12pt,onepage,oneside]{article}
\usepackage{fullpage}
\usepackage{amsfonts,amsmath,amsthm}
\usepackage{amssymb}
\usepackage{graphicx}
\usepackage{dcolumn}
\usepackage{siunitx}
\usepackage{tabularx}
\usepackage{booktabs}
\usepackage{xcolor,colortbl}
\usepackage{color}
\usepackage{comment}
\usepackage{setspace}
\usepackage{bbm}
\usepackage{natbib}
\usepackage{threeparttable}
\usepackage{multirow}
\usepackage{lscape}
\usepackage{pdflscape}
\usepackage{rotating}
\usepackage{float}
\usepackage[titletoc,toc,title]{appendix}
\usepackage{amsmath}
\usepackage{subcaption}
\usepackage{color}
\usepackage{amsopn}
\usepackage[colorlinks=true,linkcolor=cyan,allcolors=cyan]{hyperref}
\usepackage{graphicx}
\graphicspath{ {LateX} }
\usepackage{pdflscape}
\usepackage{rotating}
\usepackage{float}
\usepackage{authblk}
\usepackage{subcaption}
\usepackage[left=1in,right=1in,top=1.in,bottom=1.in]{geometry}
\usepackage{xcolor}

\usepackage{booktabs}
\usepackage{longtable}
\usepackage{array}
\usepackage{multirow}
\usepackage{wrapfig}
\usepackage{float}
\usepackage{colortbl}
\usepackage{pdflscape}
\usepackage{tabu}
\usepackage{threeparttable}
\usepackage{threeparttablex}
\usepackage[normalem]{ulem}
\usepackage{makecell}
\usepackage{xcolor}

\newcommand{\indvix} {$IVIX^{(I)}$~}

\newcommand{\spx} {$S\&P 500$}

\newcommand{\bX}{\ensuremath{\mathbf{IVIX}}}

\newcommand{\bPhi}{\ensuremath{\boldsymbol\Phi}}
\newcommand{\bPsi}{\ensuremath{\boldsymbol\Psi}}
\newcommand{\btheta}{\ensuremath{\boldsymbol\theta}}
\newcommand{\bepsilon}{\ensuremath{\boldsymbol\epsilon}}
\newcommand{\bbeta}{\ensuremath{\boldsymbol\eta}}
\newcommand{\bSigma}{\ensuremath{\boldsymbol\Sigma}}
\newcommand{\bOmega}{\ensuremath{\boldsymbol\Omega}}

\newcommand{\mC}{\ensuremath{\mathcal{C}}}

\title{\vspace{-1.6cm} Dynamic industry uncertainty networks and the business cycle\thanks{We thank Charlie Cai, Daniel Greenwald, Stephen Terry, Emil Verner and Jacky Qi Zhang for their valuable comments. We thank the participants at the 3rd Workshop on Macroeconomic Research (Cracow University), CFE 2020, and Southwestern Finance Association 2021 Annual Meeting. Mattia Bevilacqua gratefully acknowledges the support of the Economic and Social Research Council (ESRC) in funding the Systemic Risk Centre [grant number ES/K002309/1 and ES/R009724/1]. Jozef Barunik gratefully acknowledges support from the Czech Science Foundation under the EXPRO GX19-28231X project.}}

\author[1]{Jozef Barun\'{\i}k}
\author[2]{Mattia Bevilacqua}
\author[3]{Robert Faff}%
\affil[1]{Charles University and Czech Academy of Sciences}
\affil[2]{London School of Economics, Systemic Risk Centre}
\affil[3]{Bond University, Australia}
\affil[3]{University of Queensland, Australia}

\date{}

\begin{document}

\maketitle


\begin{abstract}
We argue that uncertainty network structures extracted from option prices contain valuable information for business cycles. Classifying U.S. industries according to their contribution to system-related uncertainty across business cycles, we uncover an uncertainty hub role for the communications, industrials and information technology sectors, while shocks to materials, real estate and utilities do not create strong linkages in the network. Moreover, we find that this ex-ante network of uncertainty is a useful predictor of business cycles, especially when it is based on uncertainty hubs. The industry uncertainty network behaves counter-cyclically in that a tighter network tends to associate with future business cycle contractions.
\end{abstract}



\newpage
\section{Introduction}

Throughout history, industrial structure has witnessed influential economic and financial cycles in which different sectors seem to take on a leading role.\footnote{The ascendancy of technology and telecommunications is a notable recent example. The rapidly growing internet sector accounted for \$2.1 trillion of the U.S. economy in 2018 or about 10\% of the nation's gross domestic product (GDP). Tech companies such as Apple, Google, and Amazon are leading the stock market; any little variation in their quarterly earnings or stock market prices can move the entire index. } Fluctuations in performance, valuation and interconnection have been crucially important not only for swings in financial markets, but also for the real economy and the prediction of business cycles. Despite the intensifying linkages emerging from economic activities among industries, we do not fully understand the influence of such networks on the real economy. Given this motivational backdrop, we develop forward-looking measures of industry uncertainty network structures, and we study how industry-specific shocks to option investors' expectations linked to investor fear evolve over time as well as how these shocks relate to macroeconomic activity.

Understanding the differential actions and performance of the largest firms within an economy's various industrial sectors offers key insight into the aggregate economy.\footnote{For example, \cite{gabaix2011} reports that the total sales of the top 50 firms accounted for 25\% of GDP in 2005. As another example, in December 2004, a \$24 billion one-time Microsoft dividend boosted growth in personal income from 0.6\% to 3.7\% (Bureau of Economic Analysis, January 31, 2005).}
\cite{acemoglu2012} document that significant aggregate fluctuations can originate from firm-specific microeconomic shocks or disaggregated sectors due to interconnections between different firms and sectors, functioning as a potential propagation mechanism of idiosyncratic shocks throughout the economy.
Notably, \cite{atalay2017} concludes that 83\% of the variation in aggregate output growth is attributable to idiosyncratic industry-level shocks, while \cite{gabaix2011} contends that a salient feature of business cycles is that firms and sectors comove.

To capture uncertainty connected to industrial structures (and especially driven by large firms), recent literature has exploited information about volatility or uncertainty shocks highlighting the importance of network measures to capture the propagation of volatility mechanisms \citep[e.g.][]{acemoglu2012,carvalho2013,gabaix2016,acemoglu2017,baqaee2019,herskovic2020}. \cite{carvalho2013} argue that the sector-specific ``fundamental'' microeconomic volatility has explanatory power and can serve as an early warning signal of swings in macroeconomic volatility. Further, financial uncertainty shocks and fluctuations in risk have been identified as one of the main drivers of the U.S. business cycle \citep[see][]{bloom2009,christiano2014,leduc2016,basu2017,bloom2018,ludvigson2020}.

Despite these enormous efforts, a recent surge of interest relies on historical or ex-post analysis. In contrast, we propose an ex-ante approach employing measures extracted from option prices of individual firms. We argue that information reflected in the implied second moment of firms provides an enhanced overview of the future risk and performance of economic sectors and the firms therein, shedding light on their future outlook and evaluation in a forward-looking manner through the options investors' expectations directly related to industry-linked ``fears'' \citep[e.g.][]{Diebold2014,BARUNIK2019}.

To infer such prospective information, we exploit an extensive sample of option prices to extract option-based uncertainty measures for key major firms across the full U.S. industrial landscape. We produce aggregate uncertainty measures for each industry at a daily frequency encompassing three recession periods (including the most recent Covid-19 crisis). Based on the time-varying parameter approximating models, we then construct a measure of ex-ante industry uncertainty network that reflects investor expectations of future uncertainty relevant to each industry. In doing so, we follow the work of \cite{Diebold2014,barunik2020dynamic} who argue that time-varying variance decompositions characterize well how shocks to uncertainty create dynamic networks.\footnote{Our approach is intimately connected to network node degrees, mean degrees, and connectedness measures of \cite{Diebold2012} and \cite{Diebold2014}.} We assess the varying role of industries contribution to shocks to uncertainty across the business cycles. Our analysis then quite naturally transitions to the question of how useful are such uncertainty network measures in predicting economic activity. We argue that, compared to the currently existing approaches, industry-based uncertainty network measures represent more informative (leading) propagation channels highly relevant to predict business cycles given their forward-looking feature.


The essence of our contribution is linked to two relatively recent strands of literature. The first strand is concerned with uncertainty measures and their relationship with business cycle fluctuations \citep[e.g.][]{bloom2009,bloom2014,jurado2015,ludvigson2020}.\footnote{See also \cite{bachmann2013}, \cite{bachmann2014}, \cite{christiano2014}, \cite{decker2016}, \cite{bloom2018}, and \cite{arellano2019}.} 
The second strand focuses on the role of sector-level or firm-to-firm linkages in microeconomic shocks and their relationship with the aggregate economy and changes in business conditions \citep[see, e.g.][]{gabaix2011,acemoglu2012,carvalho2013,barrot2016,acemoglu2017,baqaee2019} or the survey in \cite{carvalho2019}. Notably, connected to the second strand of literature, are studies on the role of production networks as a propagation mechanism from individual firms and/or industries to the real economy \citep[see, e.g.][]{foerster2011,digiovanni2014,ozdagli2017,atalay2017,garin2018,auer2019,lehn2020}.

However, in contrast to this extant body of work, we adopt purely market-based networks as a mechanism to dynamically study the propagation of shocks to uncertainty from industries flowing to the real economy.
In addition, a recent body of work suggests that sector-specific shocks have become more volatile relative to aggregate shocks. For instance, \cite{garin2018} develop an ``islands'' model with two sectors showing a decline in the contribution of the aggregate volatility shocks to the business cycle in favour of sector-specific shocks.
Similarly, \cite{lehn2020} show that the empirical network is dominated by a few ``investment hubs'' that produce the majority of investment goods, which are highly volatile, and have large aggregate effects on fluctuations in the business cycle. Inspired by these latest findings, our work also builds on a similar framework. Augmenting the definition of ``hubs'' from the input-output network literature, we characterize critical ``uncertainty hubs'', as industries that largely transmit and/or receive uncertainty across the business cycle, versus ``non-hubs'', being those industries that are (largely) neutral across business cycles.

A shock to an uncertainty hub may directly trigger consequences for that specific industry, for the whole system and the broader economy. While it can affect production, employment and growth at the hub, it can also generate larger uncertainty spillovers, changes in prices, growth and production of other sectors in the system, affecting the broader economy. Our proposed framework with respect to the role of uncertainty hubs in driving fluctuations in the aggregate economy is reminiscent of investment-specific technology shocks \citep[e.g.][]{greenwood2000,justiniano2010}.
Given this framework, we hypothesize that the industry network constructed from uncertainty hubs contains greater predictability for business cycles compared to the non-hubs uncertainty network.

The main findings of our paper are as follows. First, we identify industries showing a stronger (versus weaker) contribution of shocks to uncertainty and their propagation, thus playing an essential role within the aggregate industry uncertainty network. Specifically, the communications, IT and industrials play a key role, being classified as the main uncertainty hubs.
In contrast, materials, utilities and real estate are classified as non-hubs. Further, our analysis suggests that the hub role of the financial industry is limited to the global financial crisis (GFC). Second, we have further insights relevant to the time-series perspective. We find that the ex-ante industry network rises sharply during the dot com bubble, the GFC, increasing steadily afterwards. Additionally, we detect a rising importance of hubs-specific shocks in the last decade implying that shocks to hubs account for the majority of aggregate fluctuations post-GFC. Third, our empirical exercise shows that the ex-ante industry uncertainty network generates an important channel through which sector-specific shocks propagate and it is a useful leading predictor of business cycles up to one year. Changes in business cycle indicators occur in line with changes in the comovement of uncertainty across industries, and they are especially due to changes in the network connectedness in uncertainty hubs. 
These results suggest that uncertainty hubs are a major channel for targeting policy effectiveness. Our results are robust to several robustness checks and the inclusion of additional control variables.

The remainder of this paper is organized as follows. Section \ref{IndustryVIXData} describes the data and sampling used in our study.  Section \ref{connectednessmethodology} sets out the essence of the TVP-VAR network connectedness method applied to our chosen industry setting. Section \ref{findingstotal} studies the dynamic aggregate uncertainty network connectedness, and section \ref{totnetdynamic} presents the findings for the dynamic idiosyncratic uncertainty network connectedness through the business cycle. Section \ref{predicting} studies the predictive ability of the networks for the real economy. Section \ref{conclusion} concludes the paper. Additional results are relegated to the appendix of the paper.

\section{Industry uncertainty, investor beliefs and option prices}  \label{IndustryVIXData}

According to \cite{kozeniauskas2018}, conceptually three different types of uncertainty have been used in existing research: measures of uncertainty about macroeconomic outcomes (macro uncertainty); measures of the dispersion of firm outcomes (micro dispersion); and measures of the uncertainty that people have about what others believe (higher-order uncertainty). A common origin for the various uncertainty shocks can be found in macro uncertainty, shocks of which generate positive covariances between all pairs of uncertainty types.

The VIX is a common proxy used in the financial economics literature to measure the unpredictability of future aggregate outcomes, or macro uncertainty \citep[see][]{bloom2009,bekaert2013,leduc2016,kozeniauskas2018,bhattarai2020}. The VIX index (often referred to as ``fear index''), introduced by the Chicago Board Options Exchange (CBOE), is a model-free forward-looking measure implied by options prices and reflects investors' expectations about uncertainty in the stock market. It has also been linked with future macro outcomes in the aforementioned studies.


To study the dynamic uncertainty network, rather than looking at the whole U.S. stock market, we develop a forward-looking measure of uncertainty reflecting investor beliefs derived from option price data at the industry level. To capture industry uncertainty, we use forward-looking uncertainty measures that are intimately related to the VIX methodology and derived from the uncertainty of single firms within the main U.S. industries. More details on the composition of the 11 U.S. industries are provided in the Appendix, section \ref{breakdownindustries}.

According to \cite{kozeniauskas2018}, micro uncertainty describes an increase in the uncertainty that firms have about their outcomes due to changes in idiosyncratic variables. To this end, while our proposed network measure might be influenced by a common structure in volatilites, importantly our measure contains additional critical information about idiosyncratic volatilities. Indeed, single firm VIX measures can be a valid proxy for both macro and micro uncertainty, therefore capturing macro volatility as well as dispersion of firm-specific outcomes. Moreover, such micro uncertainty is challenging to measure due to the scarcity of data on firm-specific beliefs. Our measures overcome this limitation since can be computed at a high (daily) frequency. Finally, options-based measures of risk are superior to historical volatility measures with respect to both predictive power and the set of information they encompass \citep[e.g.][]{christensen1998,santa2010,BARUNIK2019}.



\subsection{Extracting model-free industry uncertainty from option prices}  \label{IndividualVIX}

For each chosen firm, we compute a model-free implied volatility index as detailed in the Appendix, section \ref{sec: appB}. This measure reflects expectations about investor uncertainty regarding the individual firm over the coming 30-day horizon. We then aggregate the individual firm information to construct a measure of ex-ante uncertainty at the industry level.
Such a measure reflects the industry expected uncertainty over the next 30 days.

More formally, the ex-ante industry uncertainty measure $\text{IVIX}^{(\text{Ind})}_t$ is constructed by taking the time-varying weighted average of the main five stocks in each industry and at each point in time through our sample period as:
\begin{equation}\label{Ind_VIX}
   \text{IVIX}^{(\text{Ind})}_t =  \sum_{s \in N^{(\text{Ind})}}\mathcal{W}_{t}^{(s)} \text{VIX}_{t}^{(s)}
\end{equation}
where $\text{Ind} \in \{1,\ldots,11\}$ represents the industry we consider, $s$ is an index for one of the $N^{(\text{Ind})}$ firms included in the given industry at time $t$, $\text{VIX}_{t}^{(s)}$ is the implied volatility for an individual stock $s$, and $\mathcal{W}_{t}^{(s)}$ is the time-varying market capitalization weight of that specific stock $s$ computed as the ratio between the time-varying market capitalization of the stock and total market capitalization of all stocks included in the industry.


\subsection{Data}  \label{Data}

We use daily data encompassing the sample period January 2000 to May 2020
in each of the following 11 U.S. industries: consumer discretionary (CD), communications (CM), consumer staples (CS), energy (E), financials (F), health care (HC), industrials (IN), information technology (IT), materials (M), real estate (RE) and utilities (U).
More specifically we merge two data sources. We use options data from OptionMetrics from January 2000 to December 2018 to compute the individual stock VIXs. We expand the coverage of the VIX time series from January 2019 until May 2020 aided by the IHS Markit's Totem Vanilla Volatility Swap data set from which we collect broker-dealers consensus prices for the volatility strike of the swaps.\footnote{The Totem database is a service within IHS Markit that gathers a large variety of derivatives marks from the major broker-dealers and returns consensus prices. In the volatility swaps service, contributors are requested to price the volatility strike at which the swap would have an inception price of zero which should be the traders best estimate of mid-market. Before 2019 January, the majority of the data in the Totem Vanilla Volatility (or Variance) Swap services were monthly, therefore would not have served our purpose.} This aids the expansion of the individual stock VIXs since we replace $\text{VIX}_{t}^{(s)}$ with the volatility strike of the swaps and we then compute the industry uncertainty measures as in equation \ref{Ind_VIX}.\footnote{On the equality between VIX and the strike of a volatility swap see, for instance, \cite{filipovic2016} and \cite{cheng2019}. For more details on the VIX index, the variance swap market and their relationship see \cite{Carr&Wu2005} and \cite{carr2009}.} The other financial information such as market capitalization and trading volume regarding the selected stocks are collected from Bloomberg.

Overall, our data set includes options prices of 69 U.S. firms. We select the largest constituent stocks in each U.S. industry at every point in time according to time-varying market capitalizations, new IPOs, exclusions of the stocks from the \spx, or missing data. For instance, for some industries such as industrials, the same five stocks have been adopted throughout the sample period. In more dynamic sectors such as IT, we observe several changes in the stocks ranking within our sample period. In cases where options on a specific firm have been only issued in recent times, we include the next ranked firm as a substitute to always ensure at least five stocks for every sector across our time period in that industry with available data.\footnote{The only exceptions are the materials and real estate industries. For the first, we use only four stocks, monthly interpolated between 01-2019 and 04-2019 due to data availability. This is because before 05-2019 for a few stocks in this sector the submission service was monthly. Between 01- and 04-2019 no single firm volatility data is available for the real estate industry, hence we replace $\text{IVIX}^{(\text{RE})}_t$ directly with the real estate sector volatility measure submitted to IHS Markit/Totem.} Table \ref{StocksInfo} in the appendix shows the included stocks within each industry and their available time period. Figure \ref{IVIX} in the appendix depicts an example of individual firm uncertainty $\text{VIX}^{(i)}$.\footnote{The CBOE has introduced stock market VIX series for a few stocks in the U.S. Comparing our calculations, with available period CBOE counterparts, show a correlation, on average, exceeding 94\%. This minor divergence is likely due to the interpolation among the two closest expiration dates to 30 days used in the CBOE methodology. For the data collected by IHS Markit spanning a shorter time frame, the correlation between the consensus volatility and CBOE VIX series is, on average, above 97\%, this again due to the interpolation used in the CBOE methodology.}

The selected stocks account for more than 58\% of the U.S. \spx~ market capitalization, thus being a valid proxy for the 11 U.S. industries, and being representative of a not trivial fraction of the U.S. GDP \citep[e.g.][]{gabaix2011}. A large representation of the U.S. stock market and its industries is what matters when studying these as economic and business cycle drivers.
 We report the descriptive statistics for the industry uncertainty indexes in Table \ref{INDVIXDescript}.

\begin{table}[hbt!]
 \centering
 \caption{Industry Uncertainty $\text{IVIX}^{(\text{Ind})}_t$ Descriptive Statistics}
  \label{INDVIXDescript}
  \begin{threeparttable}
  \footnotesize{
 \centering
\begin{tabular}{cccccccccccc}
 \toprule
$\text{Ind}$ & CD & CM & CS & E & F & HC & IN & IT & M & RE & U \\
 \toprule
Mean& 0.321&   0.318&    0.230&    0.270&    0.370&    0.252&    0.300&    0.342&    0.289&    0.343&    0.232\\
Standard Dev. & 0.116 &   0.121&    0.096&    0.106&    0.279&    0.089&    0.148&    0.144 &   0.111&    0.191&    0.107\\
Min & 0.133 &   0.1305&    0.115&    0.128&    0.145&    0.136&    0.139&    0.138&    0.152 &   0.122&    0.104\\
Max & 0.876 & 1.137 & 0.911 & 1.404 & 2.682 & 0.815 & 1.409 & 0.994 & 1.191 & 2.225 & 1.018 \\
Skewness & 1.363 &   1.538&    1.777&    3.249&    3.727&    1.552&    2.409&    1.504&    2.220&    2.769&    2.246\\
Kurtosis & 4.658&    5.755  &  6.363 &  21.240&   20.468 &   5.969 &  11.411 &   5.011&   11.247 &  17.057 &   9.788\\
\bottomrule
\end{tabular}
\begin{tablenotes}
\item {\scriptsize \textit{Notes}: This table reports the descriptive statistics for the industry uncertainty $\text{IVIX}^{(\text{Ind})}_t$ index of 11 industries: consumer discretionary (CD), communications (CM), consumer staples (CS), energy (E), financials (F), health care (HC), industrials (IN), information technology (IT), materials (M), real estate (RE) and utilities (U). The time period is from 03-01-2000 to 29-05-2020, at a daily frequency.}
\end{tablenotes}
}
\end{threeparttable}
\end{table}

From Table \ref{INDVIXDescript} we observe that the financial industry uncertainty shows the highest mean, followed by the information technology and real estate industry uncertainty measures with consumer staples and utilities showing the lowest mean values. The financial sector is also found to be the one with the highest standard deviation, and skewness of uncertainty measure. On the other hand, consumer staples, energy, health care and utilities are found to have lower standard deviations. Consumer discretionary and IT show lower skewness and kurtosis compared to the other industries uncertainty measures. The minimum values of industry uncertainty range between 10\% and 15\%, while the maximum values present a wider range with financials and real estate leading with the highest values. As an example, we plot uncertainty measures for the IT, consumer staples and financial sectors in Figure \ref{INDVIX}.

\begin{figure}[ht!]
 \begin{center}
\caption{\textbf{Industry Uncertainty}} \label{INDVIX}
\includegraphics[width=\textwidth]{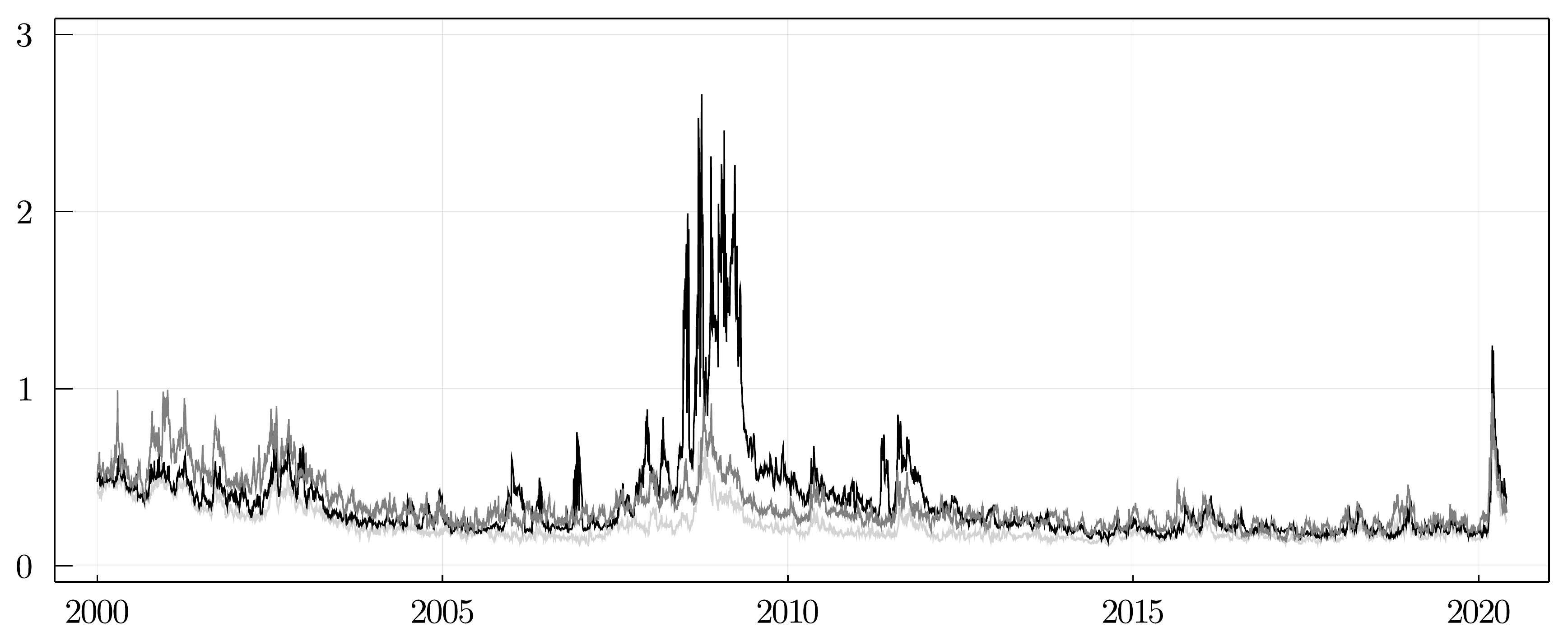}
\caption*{\textit{Notes}: This figure shows $\text{IVIX}^{(\text{Ind})}_t$ series for financials (black), IT (grey) and consumer staples (light grey) covering the period 03-01-2000 to 29-05-2020 at a daily frequency. }
\end{center}
\end{figure}

We observe that the IT sector dominates the other two during the dot com bubble and technologic boom in the early 2000s. The financial sector shows a dominant role during the GFC in 2008 and the Eurozone debt crisis in 2010 and 2011. In the most recent times, the IT sector exhibits greater uncertainty in comparison to the financial sector, showing how uncertainty might come from technology innovation and growth. Consumer staples show peaks in the early 2000s and during the GFC, however being always below the other two series and overall low throughout the sample period. All three uncertainty measures spiked in correspondence to the recent Covid-19 pandemic in March 2020. Among them, we observe that financials and IT increased during the Covid-19 crisis in a more pronounced manner than consumer staples impacted less by the Covid-19 crisis.

\section{Measurement of dynamic industry uncertainty networks}  \label{connectednessmethodology}

Industries are connected directly through counterparty risk, contractual obligations or other general business conditions of the firms. High-frequency analysis of such networks requires generally unavailable high-frequency information. In contrast, option prices and uncertainty measured in high frequencies reflect the decisions of many agents assessing risks from the existing linkages. Hence the pure market-based approach we use, in contrast to other network techniques, allows us to monitor the network on a daily frequency as well as to exploit its forward-looking strength with minimal assumptions.

Looking at how a shock to the expected uncertainty of a firm $j$ transmits to future expectations about the uncertainty of a firm $k$, we will define weighted and directed networks. Aggregating the information about such networks can provide industry level uncertainty characteristics that will measure how strongly the investors' expectations are interconnected. Importantly, we will focus on the time variation of such networks.

\subsection{Link to the network literature and causality of proposed measures}

The measures that we use are intimately related to modern network theory. Algebraically, the adjacency matrix capturing information about network linkages carries all information about the network and any sensible measure must be related to it. As noted by \cite{Diebold2014}, a variance decomposition matrix defining network adjacency matrix is then readily used as a network connectedness that is  related to network node degrees and mean degree. Currently, studies examine, almost exclusively, static networks mimicking time dynamics with estimation from an approximating window. In contrast to this approach, we follow \cite{barunik2020dynamic} who employ a locally stationary TVP-VAR that allows us to estimate the adjacency matrix for a network at each point in time with possibly large dimension. Dynamic networks defined by such time-varying variance decompositions are then more sophisticated than classical network structures in several ways.\footnote{For previous literature on the importance of the variation in the transmission of uncertainty shocks over time \citep[see][]{caggiano2014,mumtaz2018,alessandri2019}.}

In a typical network, the adjacency matrix contains a set of zero and one entries, depending on the node being linked or not, respectively. In the above notion, one interprets variance decompositions as weighted links showing the strength of the connections. In addition, the links are directed, meaning that the $j$ to $k$ link is not necessarily the same as the $k$ to $j$ link, and hence, the adjacency matrix is not symmetric. Therefore we can define weighted, directed versions of network connectedness statistics readily that include degrees, degree distributions, distances and diameters. Using the time-varying approximating model, we will define a truly time-varying adjacency matrix that will describe a dynamic network.

The proposed network connectedness measure is also directly connected to the vast economic literature about the importance of network effects in macroeconomics \citep[see][]{acemoglu2012,carvalho2013,gabaix2016,barrot2016,acemoglu2017,baqaee2019,altinoglu2020,acemoglu2020}. \cite{herskovic2020} state that measuring network effects is crucial to explain the joint evolution of firm volatility distributions. The network analysis has developed conceptual frameworks and an extensive set of tools to effectively measure interconnections among the units of a network, see for instance the survey by \cite{carvalho2019}.


Moreover, our network connectedness measure improves on shocks to uncertainty measured ex-post \citep[e.g.][]{Diebold2014}. Employing implied measures of uncertainty gives one access to a different set of information in uncertainty reflecting market participants' expectations of future movements in the underlying asset, a set of information found superior compared to ex-post measures of uncertainty \citep[see][]{christensen1998}. We are naturally interested in capturing shocks to the ex-ante uncertainty of industry $j$ that will transmit to future expectations about the uncertainty of industry $k$.\footnote{\cite{BARUNIK2019} stated that option based measures of uncertainty reflect decisions of many agents assessing the risks from the existing linkages. The options market-based approach allows us to monitor the network on a daily frequency as well as use its forward-looking strength in contrast to other network techniques based on balance sheet and other information which is generally unavailable at high frequency.}

Finally, we note that our measures can have a direct causal interpretation. \cite{rambachan2019econometric} provide an important discussion about the causal interpretation of impulse response analysis in the time series literature. In particular, they argue that if an observable time series is shown to be a potential outcome time series, then generalized impulse response functions have a direct causal interpretation. Potential outcome series describe at time $t$ the output for a particular path of treatments.

In the context of our study, paths of treatments are shocks. The assumptions required for a potential outcome series are natural and intuitive for a typical economic and/or financial time series: i) they depend only on past and current shocks; ii) series are outcomes of shocks; and iii) assignment of shocks depend only on past outcomes and shocks. The dynamic adjacency matrix we introduce in the next section is a transformation of generalized impulse response functions. Therefore, the dynamic adjacency matrix and all measures that stem from manipulations of its elements possess a causal interpretation; thus establishing the notion of causal dynamic network measures.

\subsection{Construction of dynamic uncertainty network}

To formalize the discussion, we construct a dynamic uncertainty network of industries from the industry implied volatilities computed for the main U.S. industries and we interpret the TVP-VAR model approximating its dynamics as a dynamic network following the work of \cite{barunik2020dynamic}. In particular, consider a locally stationary TVP-VAR of lag order $p$ describing the dynamics of industry uncertainty as
\begin{equation}
\label{eq:VAR2}
\bX_{t,T} = \bPhi_1 (t/T) \bX_{t-1,T} + \ldots + \bPhi_p (t/T) \bX_{t-p,T} + \bepsilon_{t,T},
\end{equation}
where $\bX_{t,T} = \left(\text{IVIX}^{(\text{1})}_{t,T},\ldots,\text{IVIX}^{(\text{N})}_{t,T}\right)^{\top}$ is a doubly indexed $N$-variate time series of industry uncertainties, $\boldsymbol \epsilon_{t,T} = \sum_{}^{-1/2}(t/T) \eta_{t,T}$ with $\eta_{t,T} \sim NID(0,I_M)$, and $\bPhi (t/T) = \Big(\bPhi_1 (t/T),....,\bPhi_p (t/T)\Big)^T$ are the time-varying autoregressive coefficients. Note that $t$ refers to a discrete time index $1\le t \le T$ and $T$ is an additional index indicating the sharpness of the local approximation of the time series by a stationary one. Rescaling time such that the continuous parameter $u \approx t/T$ is a local approximation of the weakly stationary time-series \citep{dahlhaus1996kullback}, we approximate the $\bX_{t,T}$ in a neighborhood of a fixed time point $u_0=t_0/T$ by a stationary process $\widetilde{\bX}_t (u_0)$ as
\begin{equation}
 \widetilde{\bX}_t (u_0) = \bPhi_1 (u_0) \widetilde{\bX}_{t-1}(u_0) \ldots + \bPhi_p (u_0) \widetilde{\bX}_{t-p}(u_0) + \bepsilon_{t}.
\end{equation}
The process has time-varying Vector Moving Average VMA($\infty$) representation \citep{dahlhaus2009empirical,roueff2016prediction}
\begin{equation}
\bX_{t,T} = \sum_{h=-\infty}^{\infty} \boldsymbol\Psi_{t,T,h}\boldsymbol\epsilon_{t-h}
\end{equation}
where parameter vector $\bPsi_{t,T,h} \approx \bPsi_h(t/T)$ is a time-varying impulse response function characterized by a bounded stochastic process.\footnote{Since $\bPsi_{t,T,h}$ contains an infinite number of lags, we approximate the moving average coefficients at $h=1,\ldots,H$ horizons.} The connectedness measures rely on variance decompositions, which are transformations of the information in $\bPsi_{t,T,h}$ that permit the measurement of the contribution of shocks to the system. Since a shock to a variable in the model does not necessarily appear alone, an identification scheme is crucial in calculating variance decompositions. We adapt the generalized identification scheme in \cite{pesaran1998generalized} to locally stationary processes.

The following proposition establishes a time-varying representation of the variance decomposition of shocks from asset $j$ to asset $k$. It is central to the development of the dynamic network measures since it constitutes a dynamic adjacency matrix.

\newtheorem{prop}{Proposition}
\begin{prop}[Dynamic Adjacency Matrix]\footnote{Note to notation: $[\boldsymbol A]_{j,k}$ denotes the $j$th row and $k$th column of matrix $\boldsymbol A$ denoted in bold. $[\boldsymbol A]_{j,\cdot}$ denotes the full $j$th row; this is similar for the columns. A $\sum A$, where $A$ is a matrix that denotes the sum of all elements of the matrix $A$.} \label{prop:1}
Suppose $\bX_{t,T}$ is a locally stationary process, then the time-varying generalized variance decomposition of the $j$th variable at a rescaled time $u=t_0/T$ due to shocks in the $k$th variable forming a dynamic adjacency matrix of a network is
  \begin{equation}
    \Big[ \btheta^H(u) \Big]_{j,k} = \frac{\sigma_{kk}^{-1}\displaystyle\sum_{h = 0}^{H}\Bigg( \Big[\bPsi_h(u) \bSigma(u)\Big]_{j,k} \Bigg)^2}{\displaystyle \sum_{h=0}^{H} \Big[ \bPsi_h(u) \bSigma(u) \bPsi_h^{\top}(u) \Big]_{j,j}}
  \end{equation}
where $\bPsi_h(u)$ is a time-varying impulse response function.
\end{prop}
\begin{proof}
	See Appendix \ref{app:proofs}.
\end{proof}

It is important to note that proposition \ref{prop:1} defines the dynamic network completely. Naturally, our adjacency matrix is filled with weighted links showing strengths of the connections over time. The links are directional, meaning that the $j$ to $k$ link is not necessarily the same as the $k$ to $j$ link. Therefore the adjacency matrix is asymmetric.

To characterize network uncertainty, we define total dynamic network connectedness measures in the spirit of \cite{Diebold2014,barunik2020dynamic} as the ratio of the off-diagonal elements to the sum of the entire matrix

\begin{equation}\label{network}
\mC^H(u) = 100\times\displaystyle \sum_{\substack{j,k=1\\ j\ne k}}^N \Big[\widetilde \btheta^H(u)\Big]_{j,k}\Bigg/\displaystyle \sum_{j,k=1}^N \Big[\widetilde \btheta^H(u)\Big]_{j,k}
\end{equation}
where $\Big[\widetilde \btheta^H(u)\Big]$ is a normalized $\btheta$ by the row sum. This measures the contribution of forecast error variance attributable to all shocks in the system, minus the contribution of own shocks. Similar to the aggregate network connectedness measure that infers the system-wide strengths of connections, we define measures that will reveal when an individual industry is a transmitter or a receiver of uncertainty shocks in the system. We use these measures to proxy dynamic network uncertainty. The dynamic directional connectedness that measures how much of each industry's $j$ variance is due to shocks in other industry $j\ne k$ in the economy is given by
\begin{equation} \label{from}
\mC_{j\leftarrow\bullet}^H(u) = 100\times\displaystyle \sum_{\substack{k=1\\ k\ne j}}^N \Big[\widetilde \btheta^H(u)\Big]_{j,k}\Bigg/\displaystyle \sum_{j,k=1}^N \Big[\widetilde \btheta^H(u)\Big]_{j,k},
\end{equation}
defining the so-called \textsc{from} connectedness. Note one can precisely interpret this quantity as dynamic from-degrees (or out-degrees in the network literature) that associates with the nodes of the weighted directed network we represent by the dynamic variance decomposition matrix. Likewise, the contribution of asset $j$ to variances in other variables is
\begin{equation} \label{to}
\mC_{j\rightarrow \bullet}^H(u) = 100\times\displaystyle \sum_{\substack{k=1\\ k\ne j}}^N \Big[\widetilde \btheta^H(u)\Big]_{k,j}\Bigg/\displaystyle \sum_{j,j=1}^N \Big[\widetilde \btheta^H(u)\Big]_{k,j}
\end{equation}
and is the so-called \textsc{to} connectedness. Again, one precisely interprets this as dynamic to-degrees (or in-degrees in the network literature) that associates with the nodes of the weighted directed network that we represent by the variance decompositions matrix. These two measures show how other industries contribute to the uncertainty of industry $j$, and how industry $j$ contributes to the uncertainty of others, respectively, in a time-varying fashion. Further, the \textsc{net} dynamic connectedness showing whether an industry is inducing more uncertainty than it receives from other industries in the system can be calculated as the difference between \textsc{to} and \textsc{from} is as $\mathcal{C}_{j,\textsc{net}}^H(u) =   \mathcal{C}_{j \rightarrow \bullet}^H (u)  - \mathcal{C}_{j \leftarrow \bullet}^H(u)$ and the \textsc{agg} connectedness measure as $\mathcal{C}_{j,\textsc{agg}}^H (u) =   \mathcal{C}_{j \rightarrow \bullet}^H (u)  + \mathcal{C}_{j \leftarrow \bullet}^H(u)$.

Finally, to obtain the time-varying coefficient estimates, and the time-varying covariance matrices at a fixed time point $u=t_{0}/T$, $\bPhi_{1}(u),...,\bPhi_{p}(u)$ $\bSigma(u)$, we estimate the approximating model in (\ref{eq:VAR2}) using Quasi-Bayesian Local-Likelihood (QBLL) methods \citep{petrova2019quasi}. Specifically, we use a kernel weighting function that provides larger weights to observations that surround the period whose coefficient and covariance matrices are of interest. Using conjugate priors, the (quasi) posterior distribution of the parameters of the model are available analytically. This alleviates the need to use a Markov Chain Monte Carlo (MCMC) simulation algorithm and permits the use of parallel computing. Note also that in using (quasi) Bayesian estimation methods, we obtain a distribution of parameters that we use to construct network measures that provide confidence bands for inference. We detail the estimation algorithm in Appendix \ref{app:estimate}.

\section{The dynamics of the network connectedness of industry uncertainties} \label{findingstotal}

The technological and housing market bubbles, the commodity crash, and the Covid-19 pandemic are a few major examples that show how a dramatic increase in uncertainty and different investors' expectations can rise sharply in many alternative industries. Being able to temporally and precisely characterize the industry-based network dynamics is crucial given that industries can swiftly change their characteristics and macro-economic roles. Working with dynamic network estimates, we can characterize and assess industry uncertainty connections in a timely and forward-looking manner according to the precise events leading to a more or less connected industry uncertainty network. This measurement provides new insights about the propagation of the ex-ante uncertainty shocks over different phases of the business cycle, and it identifies periods in which the U.S. industries' uncertainty was tightly connected.

We compute the dynamic aggregate network connectedness through equation \ref{network} and present its dynamics in Figure \ref{INDConnectedness}. We identify several cycles mainly driven by key events that took place in our sample such as the dot com bubble in the early 2000s, the housing market bubble, the 2007-2009 GFC, and the most recent Covid-19 crisis. Some events might be described as bursts that rapidly subside, others might be characterized by a more continuous pattern and trend. We also split the time period into inversions, recessions and expansions following the NBER classification. Inversions are marked between July 2000 and March 2001 and between September 2006 and December 2007.\footnote{The FOMC raised the target fed funds rate by 25 basis points on June 29, 2006 and lowered the target by 50 basis points on September 18, 2007. \cite{adrian2008} identified September 2006 as the end of the tightening cycle because during that month the one-month fed futures rate went from higher than the spot rate to lower. Due to the unprecedented causes of the pandemic recession in 2020, this has resulted in a downturn with different characteristics and dynamics than prior recessions. Hence, we are unable to establish an inversion period for the Covid-19 crisis that we signal only as a recession from February 2020. See also the NBER website: \url{https://www.nber.org/cycles.html}.} The recessions are marked between April 2001 and November 2001, between January 2008 and June 2009, and between February 2020 until the end of the sample, while the other years are marked as expansions.

\begin{figure}[ht!]
 \begin{center}
\caption{\textbf{Dynamic Network Connectedness of Industry Uncertainties}} \label{INDConnectedness}
\includegraphics[width=\textwidth]{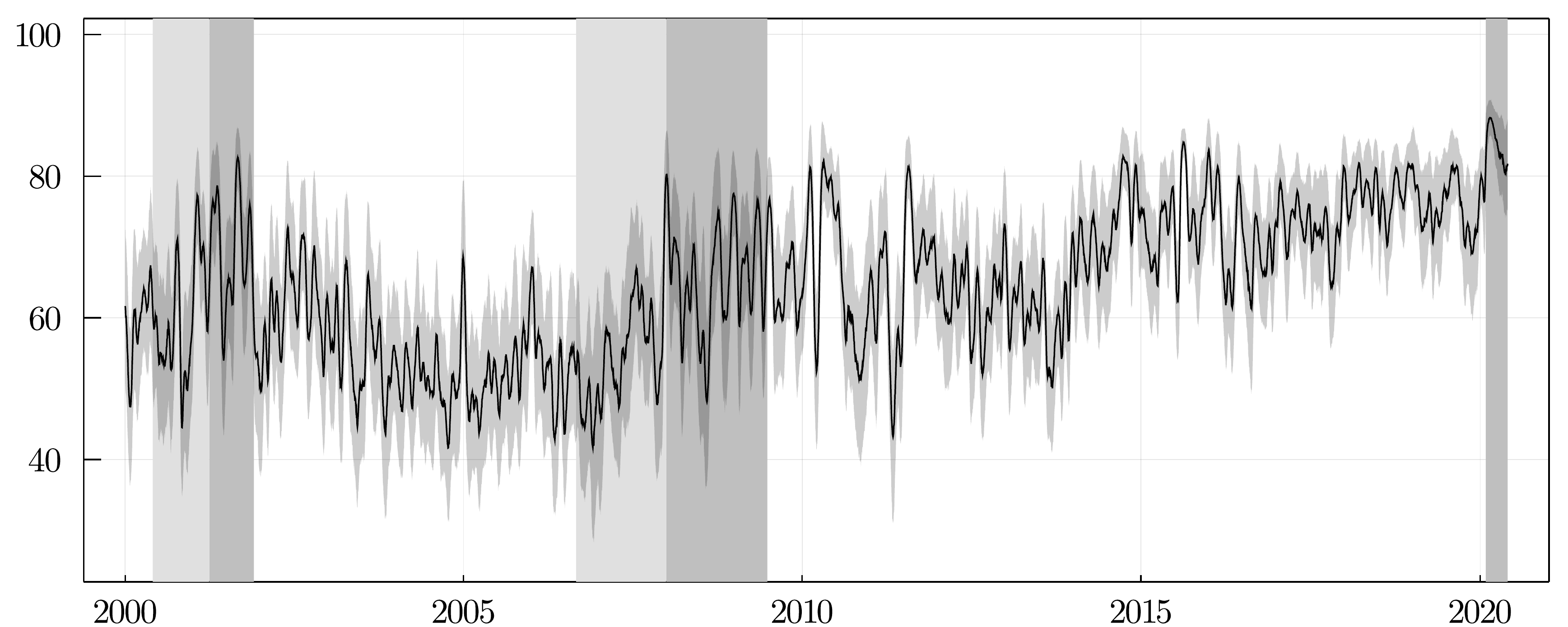}
\caption*{\textit{Notes}: This figure shows the dynamic uncertainty network connectedness with respect to the 11 U.S. industries estimated during the 03-01-2000 -- 29-05-2020 period at a daily frequency. Inversions (light grey area) are marked between July 2000 and March 2001 and between September 2006 and December 2007. The recessions (grey area) are marked between April 2001 and November 2001, between January 2008 and June 2009, and between February 2020 until the end of the sample, while the other years are marked as expansions. Note the network connectedness is plotted with two standard deviation percentiles of the measure.}
\end{center}
\end{figure}

We document the system to be strongly connected with values fluctuating around 60\% in the first half of 2000. The first cycle starts with the burst of the tech bubble in 2000, and with the network measure climbing from about 60\% to 68\%, and increasing up to about 80\% in the second half of 2001 as a response to the dot com bubble strengthening U.S. industries' uncertainty connections via shock to uncertainty from the technology industry. The index recovers to the initial level until 2004, hitting minimum values in our sample at the end of 2004 before spiking again. After that, uncertainty network connectedness shows a new lower average level, fluctuating around 50\% until the second half of 2007, the only exception being a peak at the end of 2005 which might be due to the U.S. housing bubble showing a connectedness level up to 67\%.

We observe the index recording a significant upward movement from the beginning of 2007 to 2009 reaching a level close to 80\%, in response to the high uncertainty during the 2007–2009 GFC spreading from the financial industry to other industries. Several cycles can be detected during the 2007--2009 GFC: the first between the first quarter of 2007 and August 2007 reflecting the U.S. credit crunch; the second in January--March 2008 (panic in stock and foreign exchange markets, and Bear Stearns' takeover by JP Morgan), the collapse of Lehman Brothers in September 2008 showing a spike from 47\% to 75\% in our network connectedness, and lastly in the first half of 2009 when the financial crisis started to propagate among all other industries, increasing the average network connectedness level.

Uncertainty connectedness spikes again in line with the two phases of the European sovereign debt crisis, in 2010 and the second half of 2011, reaching one of the highest levels, up to that point, close to 80\%. We then observe a drop in 2012, this being followed by a quite calm period from 2012 to mid-2013. Connectedness spikes again at the end of 2013 due to trade wars and energy turmoils reaching levels above 70\%, and twice at the end of 2014 and in mid-2015 reaching levels overcoming 80\%. Connectedness peaks in correspondence of Brexit in 2016 and at the end of 2017, and eventually in 2020 due to the coronavirus outbreak, reaching its all-time maximum value in March 2020 (a level of almost 90\%), signalling the beginning of the Covid-19 crisis. This reflects the tight connectedness among all industries in the coronavirus period since almost all industries have been severely affected.

The fluctuations of the aggregate uncertainty network across crises, market downturns and expansions
open up for a further investigation of the role of each industry uncertainty network characteristics. The U.S. industries appear to be more connected after the GFC and even more with the most recent Covid-19 crisis. We still lack an understanding of which industries drive the tightening or loosening of the network according to different business cycles. The next sections aim to clarify these points exploiting the precise time-varying estimation of the uncertainty network first, and its forward-looking properties in the last section.

\section{Uncertainty networks across the business cycles}  \label{totnetdynamic}

In addition to the aggregate network characteristics, we can classify each industry based on their expected contribution to shocks to uncertainty in the system across different phases of the business cycle. Specifically, we are interested to identify transmitters, receivers as well as industries being hubs of uncertainty. To this end, we classify industries according to the $\mathcal{C}_{\textsc{net}}^H$ and $\mathcal{C}_{\textsc{agg}}^H$ characteristics of dynamic uncertainty network.

A specific shock to uncertainty related to any industry, especially when the most influential ones are affected, can trigger major consequences for the other industries generating an aggregate impact on the whole network, tightening or weakening the uncertainty network, as well as being ultimately transmitted to the real economy. For instance, a tightening in the industry network may be connected to drops in real activity, representing a timely monitoring tool for immediate interventions by the Federal Reserve to sustain the business cycle.

\subsection{Hubs, non-hubs and business cycles}

In the case of positive or negative \textsc{net} measure, computed as the difference between \textsc{to} (equation \ref{to}) and \textsc{from} (equation \ref{from}) values, an industry is deemed to be an uncertainty transmitter or receiver, respectively. An industry receiving or transmitting shocks to uncertainty with an intermediate level can be classified as a \textit{moderate} transmitter or receiver, respectively, and may contribute to the uncertainty propagation in the system in a mild manner. An industry transmitting shocks to the system more (less) then receiving shocks from the system is labelled as a \textit{transmitter} (\textit{receiver}). An industry showing high values of both directional measures reflected by high \textsc{agg} values is playing an active role in the transmission of uncertainty shocks and is denoted as ``uncertainty hub'', and is an industry that contributes the most to uncertainty shocks within the network. Conversely, a \textit{neutral} industry showing low \textsc{agg} values is denoted as ``uncertainty non-hub''.

Industries might have changed their roles in terms of their contribution to shocks to uncertainty according to the specific economic cycle. Accordingly, we average the network characteristics across each of the three business cycle phases (inversions, recessions and expansions) as well as over the total period. Table \ref{BusinesscycleHubs} provides the details.

Financials and IT are detected as the main uncertainty hubs in the inversion and recession periods, reflecting the role of these two industries in the dot com and GFC, respectively. Also, the consumer discretionary and industrials are classified as uncertainty hubs during recessions. The IT industry is found to be the main uncertainty hub during all cycles consistently, showing a time-invariant role as a key industry in terms of the contribution of shocks to uncertainty, especially during the dot com and Covid-19 recessions. During expansion periods, the communication industry also plays a hub role in addition to the IT industry. In contrast, materials, real estate and utilities show the smallest values of network statistics and are classified as uncertainty non-hubs.

The IT, communication and industrial industries are classified as the main uncertainty hubs within the total period, with positive \textsc{net} characteristics paired with the highest \textsc{agg} values. This finding highlights the important role that the information and communication technology (ICT) industries have played in the last two decades in the system \citep[e.g.][]{jorgenson2001,bloom2012}.
Consumer discretionary and energy industries can also be classified as uncertainty hubs given their high \textsc{agg} statistics. Conversely, we find that financials, materials, real estate and utilities are overall classified as uncertainty non-hubs within the whole sample. For instance, Table \ref{BusinesscycleHubs} reveals that the IT industry, the main uncertainty hub, contributes to the network connectedness characteristics two and three-times more in comparison to the M and U industries, respectively. Market participants forward-looking expectations show that shocks to uncertainty in one of the hubs are valued differently from the corresponding shocks to uncertainty in non-hubs.

Notably, the financial industry uncertainty has been transmitting differently within different market settings, mainly during inversions and recessions, but it is overall classified as a non-hub.\footnote{Studies show how the financial and banking sector may represent a major channel in transmitting the shocks across markets during crises \citep[e.g.][]{kaminsky1999,tai2004,baur2012}.} This finding reflects the ability of authorities to influence the U.S. financial sector in the aftermath of the GFC through accommodative and unconventional monetary policies. To some, such policy interventions aimed at restoring the functioning of the financial sector during the Great Recession might have been key to avoid a second Great Depression \citep[see][]{bianchi2020}.\footnote{To note that
The IMF October 2017 Global Financial Stability Report (GFSR) finds that the global financial system continues to strengthen in response to extraordinary policy support, regulatory enhancements, and the cyclical upturn in growth.} This has also been accompanied by stronger harmonization of financial regulatory standards (e.g. the Basel capital framework). The financial sector has seen one of the most-pronounced stock market booms on record during 2009-2018. It is therefore not surprising that there has been a low level of uncertainty within the financial industry and, as a consequence, transmitted to the rest of the system.


\begin{table} [ht!]
  \centering
  \caption{\textbf{Aggregate $\mathcal{C}_{\textsc{NET}}$ and $\mathcal{C}_{\textsc{AGG}}$ across business cycles}}    \label{BusinesscycleHubs}
  \begin{threeparttable}
  \scriptsize{
  \centering
  \begin{tabular}{ccccccccccccccccc}
  \toprule
 &&  \multicolumn{3}{c}{\textbf{Inversion}} &&  \multicolumn{3}{c}{\textbf{Recession}}&&  \multicolumn{3}{c}{\textbf{Expansion}} &&  \multicolumn{3}{c}{\textbf{Total Period}}\\		
 \cmidrule{1-5}	 \cmidrule{7-9}	\cmidrule{11-13}	\cmidrule{15-17}	
  &&	  NET	&AGG & AGG $\%$   &&	 NET & AGG & AGG $\%$  && NET	&AGG  & AGG $\%$  && NET	&AGG &  AGG $\%$ \\
 \cmidrule{1-5}	 \cmidrule{7-9}	\cmidrule{11-13}	\cmidrule{15-17}	
CD&&	-1.72&	16.77 & 9.8   &&	0.49&	35.02  & 10.9&&	-0.98&	25.61 & 11.0&&	-0.71&	29.61 & 11.1\\
CM&&	-1.55&	13.07& 7.6    &&	-0.29&	28.88& 9.0&&  	1.37&	27.16& 11.6&&	1.55&	30.49& 11.5\\
CS&&	-0.87&	12.12& 7.1    &&	-0.15&	32.70& 10.2&&	-1.54&	20.93& 9.0 &&	-1.50&	24.27& 9.1\\
E&&	    -4.52&	16.88& 9.8    &&	-0.43&	30.70& 9.5&&	 -0.43&	25.09& 10.7&&	-0.80&	28.27 & 10.6\\
F&&	    3.76&	21.07& 12.3   &&	0.01&	40.55& 12.6&&	0.04&	18.02& 7.7&&	     0.54&	22.15& 8.3\\
HC&&	-0.76&	14.23& 8.3    &&	-2.15&	29.70& 9.2&&  	-0.29&	23.63& 10.1&&	-0.56&	26.45 & 9.9\\
IN&&	 0.92&	20.63& 12.1   &&	-0.21&	40.09& 12.5&&	0.26&	26.71& 11.4&&	0.17&	31.47 & 11.9\\
IT&&	4.42&	25.09& 14.7   &&	2.77&	39.74& 12.4&&	1.93&	28.48& 12.2&&	2.23&	33.71 & 12.7\\
M&&	    -0.33&	10.15& 5.9    &&	-1.77&	15.99& 5.0&&	-1.01&	14.62& 6.3&&	    -0.98&	15.84 & 6.0\\
RE&&	1.24&	12.01& 7.1    &&	1.55&	17.18& 5.3&&	0.97&	13.89& 6.0&&	    0.51&	14.22 & 5.3\\
U&&	    -0.58&	9.11& 5.3     &&	0.19&	11.04& 3.4&&	-0.24&	9.41& 4.0&&	    -0.43&	9.59 & 3.6\\
\bottomrule
 \end{tabular}
  \caption*{\scriptsize \textit{Notes}:  The table shows the average \textsc{net} and \textsc{agg} values with respect to the 11 U.S. industries' uncertainty network. When the \textsc{net} measure is positive an industry can be classified as a \textsc{net} marginal transmitter, while, when negative, it can be classified as a \textsc{net} marginal receiver. The highest values of the \textsc{agg} network statistics are associated with uncertainty hubs, while the lowest with uncertainty non-hubs. The statistics are reported for the business cycle main phases, namely inversion, recession, expansion, aggregated, and also for the total period, namely from 03-01-2000 to 29-05-2020, at a daily frequency.
  }
  }
  \end{threeparttable}
\end{table}

After having classified the industries based on their contribution to uncertainty across business cycles, we compute separate forward-looking networks extracted from uncertainty hubs only and uncertainty non-hubs only. To this end, we input the $\text{IVIX}^{(\text{Ind})}_t$ of hubs only (consumer discretionary, communications, energy, industrials and IT) and of non-hubs only (financials, materials, real estate and utilities) in the network model in equation \ref{network}.

Figure \ref{hubsplot} plots the dynamics of both network connectedness characteristics. We observe that the uncertainty network extracted from hubs shows a higher degree of integration compared to the ones extracted from non-hubs. The difference in uncertainty produced by hubs versus non-hubs becomes stronger in the post-GFC and more recent years, implying that shocks in hubs such as ICT industries and industrials have played a key role increasingly contributing to shocks in uncertainty within the system and aggregate fluctuations over time. This finding is consistent with the development of the ICT industries in the last decades. Such hubs are in fact industries providing services to other industries, sharing a role in propagating uncertainty shocks to other sectors \citep[see][]{bloom2012}. This channel can also be because uncertainty hubs are sectors likely to have financial importance and a large and fast-growing market capitalization. Thus, our proposed framework for the role of uncertainty hubs in driving fluctuations in the aggregate economy is reminiscent of investment-specific technology shocks \citep[e.g.][]{greenwood2000,justiniano2010}.

\begin{figure}[ht!]
 \begin{center}
\caption{\textbf{Uncertainty Hubs and non-Hubs Networks}} \label{hubsplot}
\includegraphics[width=\textwidth]{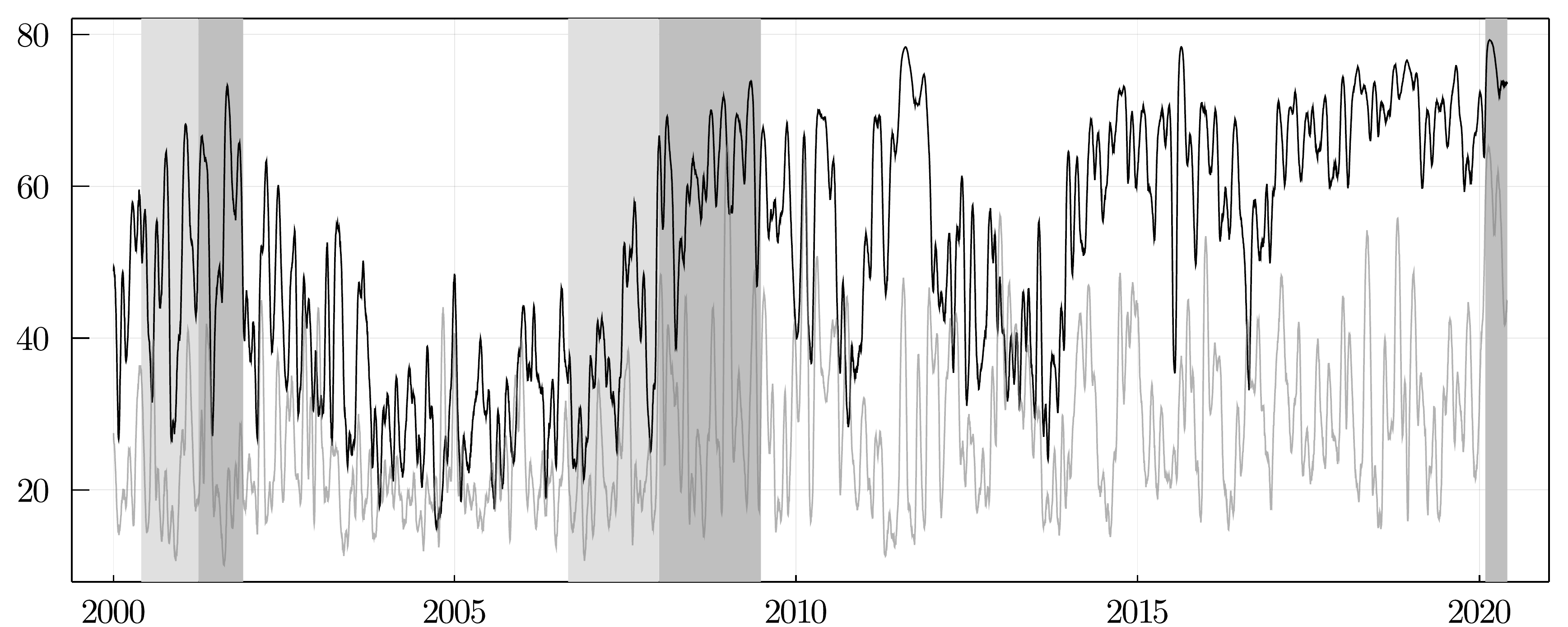}
\caption*{\textit{Notes}: This figure shows the network connectedness measures extracted from uncertainty hubs (black line) and non-hubs (grey line) from 03-01-2000 to 29-05-2020, at a daily frequency. Inversions (light grey area) are marked between July 2000 and March 2001 and between September 2006 and December 2007. The recessions (grey area) are marked between April 2001 and November 2001, between January 2008 and June 2009, and between February 2020 until the end of the sample, while the other years are marked as expansions. }
\end{center}
\end{figure}

An uncertainty shock can affect production, employment and growth within the hub, it can also generate larger uncertainty spillovers, changes in prices, growth and production of other sectors in the system, affecting the broader economy \cite[e.g.][]{kozeniauskas2018}. As an example of a few possible mechanisms, \cite{lehn2020} state that a positive shock to an investment hub directly increases production and employment in that hub; because the shock also raises the supply of investment goods, other sectors increase employment to produce more intermediate inputs for the hub. In contrast, a shock in a non-hub has a small effect on investment supply generating smaller spillovers to the rest of the economy. Further, according to the islands' framework in  \cite{garin2018}, a shock that simultaneously affects both hubs and non-hubs (both islands) is significantly weaker compared to shocks in both islands separately, this due to a reallocative shock mechanism.
In our context, this would directly translate into an increase in uncertainty of both islands (hubs and non-hubs) with the difference that non-hubs receive shocks to uncertainty, while hubs (mainly providers of services cross-sectorally) both receive and transmit shocks to uncertainty to a much larger extent (captured by a higher \textsc{agg} statistic in our model). The uncertainty transmitted or received by hubs is found to be approximately three times higher than for non-hubs.

Given the framework and rationale that we put forward, we hypothesize that the response of business cycles to shocks should be larger for hubs than for non-hubs. Shocks to uncertainty hubs generate grater effects within the network of industries and extending their implications towards the real economy. Investors' expectations and beliefs generating shocks to uncertainty in hubs reflect value-added growth into those, future price increases or losses, productivity shocks and changes in future outcomes which can trigger economic booms or downturns to a larger extent than in non-hubs.
Therefore, the hubs-based network measure is the natural candidate to be a better leading indicator of business cycles compared to non-hubs, which is explored in the next section.

\subsection{Shocks to uncertainty in specific business cycles}

Before moving on, we briefly show a more granular classification of industries for each business cycle. We report the $\mathcal{C}_{\textsc{net}}$ and $\mathcal{C}_{\textsc{agg}}$ statistics in Table \ref{CumNetBusinesscylce} in the appendix.

The IT industry can be classified as the main uncertainty hub during both inversion and recession phases related to the dot com bubble (from 2000 to mid-2002), contributing to a spread in uncertainty within the whole system. Given the high \textsc{agg} values, CD and F are classified as main uncertainty hubs during the dot com recession. In contrast, M and U are classified as uncertainty non-hubs during this first recession. Interestingly, around 2002-2004 we observe a role of uncertainty transmitter for RE since this period corresponds to the U.S. housing market bubble. The 2007--2008 GFC was indeed related to the bursting of a real estate bubble, identified precisely by the dynamic network at the end of 2005.

We then observe a main role as \textsc{net} uncertainty transmitter for F between 2006 and 2007 reflecting the first events related to the GFC and the mortgage crisis, and especially in correspondence to the collapse of Lehman Brothers in September 2008. F shows a predominant role during the GFC, both in inversion and actual recession, contributing to spread uncertainty across the whole system and being therefore classifiable as uncertainty hub within this period. After this, the U.S. economy is characterized by a long recovery period in which the ICT industries are, by far, the main uncertainty hubs. From 2010 until 2015, we observe E taking on a role as an uncertainty hub, and this could be due to a combination of events throughout this time period,\footnote{For instance the changes in the U.S. energy policies, U.S. conflicts in the Middle East, OPEC excessive production, and U.S. oil prices collapse leading to higher oil price volatility. Moreover, in summer 2011, oil and other commodities prices have fallen, although from historically high levels, and this was reflected in the fear that the commodities boom was over, with raw materials set to drop sharply.} emphasized by the spike in uncertainty between June 2014 and February 2015 associated with the global commodity price crash and oil price drop. From 2015 until the end of our sample period, industrials, consumer discretionary, communications and IT are detected as the main uncertainty hubs within our system, whereas all the other industries either as mild receivers or as uncertainty non-hubs. Finally, in the most recent Covid-19 recession, we note that the same industries are found to be the main uncertainty hubs, while health-care, real estate and utilities as uncertainty non-hubs.

This analysis further highlights the usefulness of the time-varying parameter model to more precisely uncover timing in the role of industries in shocks to uncertainty. Business cycles, economic downturns and crises with different nature might intensify the role of a specific industry, becoming a key player in contributing to shocks to uncertainty within the network. It is then crucial to be able to identify the role of each industry across time. Moreover, differences in terms of shocks to uncertainty might play a critical role in predicting business cycles throughout the sample.

\section{Industry uncertainty networks and business cycle predictability}   \label{predicting}

We study whether the information extracted from options of large firms and aggregated into ex-ante industry uncertainty networks may contribute to the predictability of indicators of the business cycle.
We draw from previous literature on the role of aggregated sector-level or firm-level networks of \textit{microeconomic} shocks to uncertainty and their relationship with the aggregate economy and business conditions \citep[e.g.][]{gabaix2011,acemoglu2012,carvalho2013,barrot2016,atalay2017,lehn2020}.
The forward-looking aspect of our network is what we are interested in exploiting in this section to predict future business cycles in advance. Further, we also argue that network measures extracted from uncertainty hubs play a more informative predictive role compared to the ones extracted from uncertainty non-hubs.

The literature on business cycle indicators, turning points, expansions and recessions characterizations has a long history. The Business Cycle Dating Committee (BCDC) of the National Bureau of Economic Research (NBER) provides probably the largest historical record of the economic activity classification and business cycles. However, its releases are usually made with about a year's delay. It, therefore, represents a very reliable chronology of peaks and troughs throughout history rather than providing an early warning tool. The latter is what studies have been aiming to propose for decades. We study the relationship between uncertainty network connectedness and more timely indicators of business cycles, both coincident and leading indicators. We hypothesize that our network connectedness may represent an even more timely and forward-looking predictor of business cycle indicators given its ex-ante characteristics from the options market. It may therefore represent both a good predictor of coincident and leading indicators and it can be also classified as a leading monitoring tool of the business cycle itself. This would provide researchers, policy-makers and the public with an even more timely indicator than the ones already available.


\subsection{The predictability of business cycle coincident indicators}  \label{coincidentpredict}

\cite{berge2011} provides an exhaustive summary of the different measures available that provide reliable signals about the current state of the business cycle. Among those, we adopt the Chicago Fed National Activity Index (CFNAI) and the Aruoba, Diebold, and Scotti (ADS) index of business conditions \citep[see][]{aruoba2009}.\footnote{The CFNAI is a monthly index that tracks the overall economic activity and the inflationary pressure. It is computed as the first principal component of 85 series drawn from four broad categories of data all adjusted for inflation. A zero value for the monthly index has been associated with the national economy expanding at its historical trend (average) rate of growth; negative values with below-average growth; positive values with above-average growth. See also \cite{chava2020} on the adoption of the CFNAI as business cycle indicator. For more information see \url{https://www.chicagofed.org/publications/cfnai/index}. The Aruoba-Diebold-Scotti (ADS) Business Condition Index tracks real business conditions at a high frequency and it is based on economic indicators. The average value of the ADS index is zero. Progressively positive values indicate progressively better-than-average conditions, whereas progressively more negative values indicate progressively worse-than-average conditions. It is collected from: \url{https://www.philadelphiafed.org/research-and-data/real-time-center}.} Specifically, we adopt the 3-month moving average of the Chicago FED National Activity Index (CFNAI-MA3) and we aggregate the ADS indicator at a monthly frequency. We adopt these as business cycle coincident indicators and we study whether the uncertainty network can forecast them.

We also disentangle these business cycles indicators into proxies of expansions and recessions following the decomposition approach by \cite{berge2011} who proposed optimal thresholds for CFNAI and ADS equal to -0.72 and -0.80, respectively. Thus, periods of economic expansion are associated with values of the CFNAI-MA3 (ADS) above –0.72 (-0.80), whereas periods of economic contraction with values of the CFNAI-MA3 (ADS) below –0.72 (-0.80). We study whether the predictive ability of the uncertainty network connectedness varies according to different states of the business cycle. We aggregate the network connectedness measure at a monthly frequency to match the frequency of the business cycle indicators and we run the following predictive regression:
\begin{equation} \label{PredictregresstotControl}
\mathcal{Y}^{(\ell)}_{t+h} = \beta_0 + \beta_{\mathcal{C}}  \  \mathcal{C}_t +  \sum_{i=1}^{N} \beta_{X,i} \ X_{t,i}  + \epsilon_t
\end{equation}
where $\mathcal{Y}^{(\ell)}_{t+h}$ is one of the business cycle indicators  we select (or their components) with the predictive horizon $h \in {1,3,6,9,12}$ months. The $\mathcal{C}_t$ is the industry uncertainty network measure (note we drop index $H$ here for the ease of notation), $X_{t,i}$ is a set of control variables including both traditional predictors of business cycles such as oil price changes (OIL), term-spread as 10-year bond rate minus the 3-month bond rate (TS), unemployment rate (UR), \citep[see also][]{gabaix2011}, and also a potential leading indicator extracted from the financial markets, namely the changes in the CBOE VIX index being a common proxy for macro uncertainty in the U.S. (VIX), changes in the \spx\ price index (SPX), the Bloomberg Commodity price index (COMM) and the S\&P Case-Shiller Home price (CSHP).\footnote{Oil prices, 10-year and 3-month bond rates, unemployment rate and the S\&P Case-Shiller Home price index are collected from the Federal Reserve Bank of St. Louis economic database at \url{https://fred.stlouisfed.org/}; the CBOE VIX index, \spx\ price index and the Bloomberg Commodity price index are collected from Bloomberg.} Therefore, $X_{t,i}$ is indexed for $i$ up to $N = 7$, the number of control we select, with $i \in (\text{OIL, TS, UR, VIX, SPX, COMM, CSHP})$.

Table \ref{CFNAI-MA3Pre} reports the predictive results. We observe that the uncertainty network is a strong predictor of the aggregate CFNAI-MA3 indicator of the business cycle up to 12 months in advance also after taking into account the information of the selected controls. The coefficient associated with our independent predictor is negative suggesting that a tighter network of industry uncertainties would lead to a contraction in the business cycle in the future horizons. The network is therefore found to behave counter-cyclically, a finding in line with previous studies relating uncertainty measures with the business cycles \citep[e.g.][]{bloom2018}. The performance of the models measured by the adjusted-$R^2$ is found to be close to 30\% at the 1-month horizon, then it decreases at the semi-annual horizon, increasing again at longer horizons such as 9 and 12 months.

\begin{table}[h!]
 \centering
 \caption{CFNAI-MA3 Predictive Results} \label{CFNAI-MA3Pre}
 \begin{threeparttable}
 \centering
\scriptsize{
\begin{tabular}{clcccccc}
 \toprule
&&  \multicolumn{5}{c}{Panel A: CFNAI-MA3}\\
\toprule
&&  \multicolumn{1}{c}{\textit{h=1}} &  \multicolumn{1}{c}{\textit{h=3}} & \multicolumn{1}{c}{\textit{h=6}} & \multicolumn{1}{c}{\textit{h=9}} & \multicolumn{1}{c}{\textit{h=12}}\\
\toprule
$\mathcal{C}_t$ $\mid$ $X_t$  & & -0.013** & -0.024*** & -0.026*** & -0.040*** & -0.028*** \\
  & &(0.006) & (0.007) & (0.007) & (0.007) & (0.007) \\
Adj. $R^{2}$ & & \multicolumn{1}{c}{0.294} & \multicolumn{1}{c}{0.156} & \multicolumn{1}{c}{0.067} & \multicolumn{1}{c}{0.176} & \multicolumn{1}{c}{0.168}\\
Obs &  & \multicolumn{1}{c}{243} & \multicolumn{1}{c}{241} & \multicolumn{1}{c}{238} & \multicolumn{1}{c}{235} & \multicolumn{1}{c}{232} \\
 \toprule
&&  \multicolumn{5}{c}{Panel B: CFNAI-MA3 Expansion}\\
\toprule
 &&  \multicolumn{1}{c}{\textit{h=1}} &  \multicolumn{1}{c}{\textit{h=3}} & \multicolumn{1}{c}{\textit{h=6}} & \multicolumn{1}{c}{\textit{h=9}} & \multicolumn{1}{c}{\textit{h=12}}\\
\toprule
 $\mathcal{C}_t$ $\mid$ $X_t$   && -0.007*** & -0.010*** & -0.012*** & -0.011*** & -0.008*** \\
  & & (0.002) & (0.002) & (0.002) & (0.002) & (0.002) \\
Adj. $R^{2}$ & & \multicolumn{1}{c}{0.093} & \multicolumn{1}{c}{0.175} & \multicolumn{1}{c}{0.268} & \multicolumn{1}{c}{0.304} & \multicolumn{1}{c}{0.239} \\
Obs &  & \multicolumn{1}{c}{243} & \multicolumn{1}{c}{241} & \multicolumn{1}{c}{238} & \multicolumn{1}{c}{235} & \multicolumn{1}{c}{232} \\
\toprule
&&  \multicolumn{5}{c}{Panel C: CFNAI-MA3 Recession}\\
\toprule
 &&  \multicolumn{1}{c}{\textit{h=1}} &  \multicolumn{1}{c}{\textit{h=3}} & \multicolumn{1}{c}{\textit{h=6}} & \multicolumn{1}{c}{\textit{h=9}} & \multicolumn{1}{c}{\textit{h=12}}\\
\toprule
 $\mathcal{C}_t$ $\mid$ $X_t$  & & -0.006 & -0.014** & -0.013** & -0.029*** & -0.020*** \\
  & & (0.006) & (0.007) & (0.007) & (0.007) & (0.007) \\
Adj. $R^{2}$ & &   \multicolumn{1}{c}{0.287} & \multicolumn{1}{c}{0.135} & \multicolumn{1}{c}{0.059} & \multicolumn{1}{c}{0.144} & \multicolumn{1}{c}{0.148} \\
Obs &  & \multicolumn{1}{c}{243} & \multicolumn{1}{c}{241} & \multicolumn{1}{c}{238} & \multicolumn{1}{c}{235} & \multicolumn{1}{c}{232} \\
\bottomrule
\end{tabular}
\begin{tablenotes}
\item {\scriptsize \textit{Notes}: This table presents the results of the predictive regression in equation \ref{PredictregresstotControl} between the industry uncertainty network connectedness and the 3-month moving average of the Chicago FED National Activity Index (CFNAI-MA3), indicator of business cycle (Panel A). In Panel B and Panel C the results of the predictive regression with respect to the CFNAI-MA3 expansion and recession indicators are reported, respectively. We also add a set of controls, $X$. The five columns of the table represent different predictability horizons with $h \in (1,3,6,9,12)$. Regressions' coefficients and standard errors (in parentheses), and adjusted-$R^2$ are reported. Coefficients are marked with *, **, *** for 10\%, 5\%, 1\% significance levels, respectively. Intercept and controls results are not reported for the sake of space. Series are considered at a monthly frequency between 01-2000 and 05-2020.}
\end{tablenotes}
}
\end{threeparttable}
\end{table}

When we look at the disentangled components of the business cycle indicator, we observe that uncertainty network can predict well future expansion periods up to one year in advance. The sign associated with the models' coefficients is, again, negative therefore suggesting that expansion periods might contract when the network is tighter. We notice a greater adjusted-$R^2$ performance of the model for the 6- to 12-month horizons. Regarding U.S. recession periods, we observe weaker predictability of network across shorter horizons. However, the uncertainty network is still found to predict well future recessions from 3-month up to the 12-month horizon, being the higher adjusted-$R^2$ placed again on the longer horizon. Interestingly, we find a negative sign associated with the coefficients, this implying that increasing levels of network connectedness will expand the business cycle when in recession (in this case the dependent variable is below the -0.72 threshold).\footnote{The Chicago Fed suggested -0.7 to be a more accurate threshold of turning points for the CFNAI indicator. We repeat the empirical analysis of this section by adopting this threshold. The results are found to be materially the same.}

We validate the predictive ability of uncertainty network connectedness by showing how it can also similarly predict a different business cycle coincident indicator. We show the results for the ADS Index in Table \ref{ADSPre} in the appendix. The aggregate uncertainty network measure shows predictability power for the ADS index from 3 months up to 12 months in advance, with again negative coefficients and stronger performance at the long horizon. When we look at the expansion or recession indicators, a stronger predictive ability is still placed at longer horizons, especially for the expansion indicator.

Another possible business cycle coincident indicator is the industrial production (IP) growth rate. As a robustness check, we repeat the same exercise with the IP annualized growth rate, still confirming the predictive ability of the uncertainty network. We also adopt another business cycle coincident indicator collected from the Economic Cycle Research Institute (ECRI), the U.S. coincident indicator (U.S.CI).\footnote{For more information and data see https://www.businesscycle.com/ecri-reports-indexes/all-indexes.} We take the growth rate of the indicator and show that this leads to similar results, relegated to the appendix in Table \ref{IPU.S.CIPre}. As an additional robustness check, we also replace the uncertainty network connectedness measure constructed with time-varying networks with a network measure constructed by following previous studies \citep[e.g.][]{Diebold2012,Diebold2014} using a moving window. We find that the latter is unable to predict future business cycles and the expansion and recession components, highlighting, even more, the importance of precisely characterizing the network at any point in time without relying on moving windows when it comes to predicting future levels of the business cycle or the real economy.\footnote{The all set of results is available from the authors upon request.}

\subsection{Industry uncertainty networks and leading indicators}  \label{leadingpredict}

Due to its forward-looking nature, we argue that the uncertainty network connectedness could also potentially serve as a good predictor of business cycle leading indicators, and be considered a leading indicator itself. We here check the predictive ability of this index for two business cycle leading indicators: the U.S. composite leading indicator (CLI) by the OECD.\footnote{The composite leading indicator is collected from the OECD data base at \url{https://data.oecd.org/leadind/composite-leading-indicator-cli.htm}. CLI provides early signals of turning points in business cycles showing fluctuations of the economic activity around its long term potential level.} and the U.S. leading indicator (U.S.LI) computed by the Economic Cycle Research Institute (ECRI).
U.S.LI is available at a weekly frequency and aggregated here at a monthly frequency, and we take the growth rate of the indicator. We find similar results, confirming both significance and coefficients signs, even after adding the set of controls. We repeat the same predictive exercise of the previous subsection, by running equation \ref{PredictregresstotControl} where now the dependent variable is CLI. We report the results in Table \ref{CLIPre}. We observe that the predictability of uncertainty network connectedness is even stronger for the business cycle leading indicators, spanning from 3-month up to one year and from 1-month up to 9-month horizons, for CLI and U.S.LI, respectively. The coefficients are still found to be negative confirming our previous findings.

\begin{table}[h!]
 \centering
 \caption{Leading Indicators Predictive Results} \label{CLIPre}
 \begin{threeparttable}
 \centering
\scriptsize{
\begin{tabular}{clcccccc}
 \toprule
&&  \multicolumn{5}{c}{Panel A:  CLI}\\
\toprule
  &&  \multicolumn{1}{c}{\textit{h=1}} &  \multicolumn{1}{c}{\textit{h=3}} & \multicolumn{1}{c}{\textit{h=6}} & \multicolumn{1}{c}{\textit{h=9}} & \multicolumn{1}{c}{\textit{h=12}}\\
\toprule
$\mathcal{C}_t$ $\mid$ $X_t$  & & -0.011 & -0.030*** & -0.051*** & -0.079*** & -0.084*** \\
  & & (0.010) & (0.011) & (0.012) & (0.011) & (0.011) \\
Adj. $R^{2}$ & & \multicolumn{1}{c}{0.435} & \multicolumn{1}{c}{0.294} & \multicolumn{1}{c}{0.192} & \multicolumn{1}{c}{0.247} & \multicolumn{1}{c}{0.262} \\
Obs &  & \multicolumn{1}{c}{243} & \multicolumn{1}{c}{241} & \multicolumn{1}{c}{238} & \multicolumn{1}{c}{235} & \multicolumn{1}{c}{232} \\
\toprule
&&  \multicolumn{5}{c}{Panel B:  U.S.LI}\\
\toprule
  &&  \multicolumn{1}{c}{\textit{h=1}} &  \multicolumn{1}{c}{\textit{h=3}} & \multicolumn{1}{c}{\textit{h=6}} & \multicolumn{1}{c}{\textit{h=9}} & \multicolumn{1}{c}{\textit{h=12}}\\
  \toprule
 $\mathcal{C}_t$ $\mid$ $X_t$ & &  -0.348*** & -0.371*** & -0.417*** & -0.372*** & -0.108 \\
  & & (0.063) & (0.066) & (0.066) & (0.067) & (0.073) \\
Adj. $R^{2}$ & & \multicolumn{1}{c}{0.364} & \multicolumn{1}{c}{0.304} & \multicolumn{1}{c}{0.303} & \multicolumn{1}{c}{0.299} & \multicolumn{1}{c}{0.178} \\
Obs &  & \multicolumn{1}{c}{243} & \multicolumn{1}{c}{241} & \multicolumn{1}{c}{238} & \multicolumn{1}{c}{235} & \multicolumn{1}{c}{232} \\
\bottomrule
\end{tabular}
\begin{tablenotes}
\item {\scriptsize \textit{Notes}: This table presents the results of the predictive regression in equation \ref{PredictregresstotControl} between the industry uncertainty network connectedness and two leading indicators of business cycle, namely CLI and U.S.LI, in Panel A and B, respectively. We also add a set of controls, $X$. The five columns of the table represent different predictability horizons with $h \in (1,3,6,9,12)$. Regressions' coefficients and standard errors (in parentheses), and adjusted-$R^2$ are reported. Coefficients are marked with *, **, *** for 10\%, 5\%, 1\% significance levels, respectively. Intercept and controls results are not reported for the sake of space. Series are considered at a monthly frequency between 01-2000 and 05-2020. }
\end{tablenotes}
}
\end{threeparttable}
\end{table}

Overall, it appears that the uncertainty network can anticipate what is commonly viewed as a business cycle leading indicator. This opens up some interesting considerations. Given that the uncertainty network is extracted from options prices, it is expected that the newly proposed uncertainty network contains forward-looking information that can be useful as ex-ante business cycle monitoring indicator. We know that a business cycle leading indicator should ideally anticipate and predict coincident indicators. We illustrate in the previous section that the uncertainty network shares such properties. In this subsection, we also show how our network measure is a good predictor of leading indicators, such a finding emphasizing even further the usefulness of its forward-looking information content.

\begin{table}[h!]
 \centering
 \caption{Coincident Indicators Predictive Results} \label{CFNAIPreCLIcontrol}
 \begin{threeparttable}
 \centering
\scriptsize{
\begin{tabular}{clcccccc}
\toprule
&&  \multicolumn{5}{c}{Panel A:  CFNAI-3M}\\
\toprule
&&  \multicolumn{1}{c}{\textit{h=1}} &  \multicolumn{1}{c}{\textit{h=3}} & \multicolumn{1}{c}{\textit{h=6}} & \multicolumn{1}{c}{\textit{h=9}} & \multicolumn{1}{c}{\textit{h=12}}\\
\toprule
$\mathcal{C}_t$ $\mid$ $X_t$  & & -0.011** & -0.023*** & -0.025*** & -0.040*** & -0.028*** \\
  & & (0.005) & (0.006) & (0.007) & (0.007) & (0.007) \\
CLI & & 0.432*** & 0.357*** & 0.265*** & 0.136*** & 0.008 \\
  & & (0.036) & (0.045) & (0.051) & (0.051) & (0.053) \\
Adj. $R^{2}$ & & \multicolumn{1}{c}{0.560} & \multicolumn{1}{c}{0.333} & \multicolumn{1}{c}{0.162} & \multicolumn{1}{c}{0.198} & \multicolumn{1}{c}{0.164} \\
Obs &  & \multicolumn{1}{c}{243} & \multicolumn{1}{c}{241} & \multicolumn{1}{c}{238} & \multicolumn{1}{c}{235} & \multicolumn{1}{c}{232} \\
\toprule
&&  \multicolumn{5}{c}{Panel B: ADS}\\
\toprule
&&  \multicolumn{1}{c}{\textit{h=1}} &  \multicolumn{1}{c}{\textit{h=3}} & \multicolumn{1}{c}{\textit{h=6}} & \multicolumn{1}{c}{\textit{h=9}} & \multicolumn{1}{c}{\textit{h=12}}\\
\toprule
$\mathcal{C}_t$ $\mid$ $X_t$  & &  -0.018 & -0.026* & -0.034* & -0.053*** & -0.038** \\
 &  & (0.015) & (0.018) & (0.018) & (0.018) & (0.019) \\
CLI & & 0.590*** & 0.587*** & 0.425*** & 0.305** & 0.159 \\
  & & (0.109) & (0.130) & (0.134) & (0.135) & (0.138) \\
Adj. $R^{2}$ & & \multicolumn{1}{c}{0.351} & \multicolumn{1}{c}{0.105} & \multicolumn{1}{c}{0.066} & \multicolumn{1}{c}{0.080} & \multicolumn{1}{c}{0.074} \\
Obs &  & \multicolumn{1}{c}{243} & \multicolumn{1}{c}{241} & \multicolumn{1}{c}{238} & \multicolumn{1}{c}{235} & \multicolumn{1}{c}{232} \\
\bottomrule
\end{tabular}
\begin{tablenotes}
\item {\scriptsize \textit{Notes}: This table presents the results of the predictive regressions between the industry uncertainty network connectedness, and the coincident indicators of business cycle, namely CFNAI and ADS. We present results for regression equation \ref{PredictregresstotControl} in which we add a set of controls including also the leading indicator, CLI. The five columns of the table represent different predictability horizons with $h \in (1,3,6,9,12)$. Regressions' coefficients and standard errors (in parentheses), and adjusted-$R^2$ are reported. Coefficients are marked with *, **, *** for 10\%, 5\%, 1\% significance levels, respectively. Intercept and controls results are not reported for the sake of space, the only exception being the CLI control. Series are considered at a monthly frequency between 01-2000 and 05-2020.  }
\end{tablenotes}
}
\end{threeparttable}
\end{table}

To validate this point, we test whether the existing business cycle leading indicators might contain a different set of information, mainly at shorter horizons, compared to our uncertainty network. To this end, we test whether our measure can predict coincident indicators even after controlling for a leading indicator (CLI). The results are reported in Table \ref{CFNAIPreCLIcontrol}. We find that the uncertainty network predictability holds at every horizon, even after controlling for CLI. The latter shows a good predictive ability, however up to the 9-month horizon (in line with the index characteristics description). The uncertainty network clearly shows characteristics of a complementary (and rather superior) business cycle leading indicator, spanning predictive power from 1- to the 12-month horizon in advance, even after controlling for CLI. For the ADS, as a proxy for a business cycle coincident indicator, we find a weaker predictive power for the uncertainty network at the short horizon, however still confirming the anticipatory property of about one quarter. We repeat the same exercise of Table \ref{CFNAIPreCLIcontrol} by adopting the U.S.LI indicator by the ECRI as control. We obtain similar findings for both CFNAI and ADS and results are relegated to the paper appendix in Table \ref{CFNAIPreU.S.CIcontrol}. Overall, the uncertainty network adds quite a lot in terms of long-horizon predictability compared to the information content of other leading indicators of business cycles.

\subsection{Hubs and non-hubs industry connectedness networks}  \label{hubspredictive}


In this subsection, we check whether the predictability power of uncertainty hubs-based networks may differ from uncertainty non-hubs. We repeat the empirical analysis of the previous section, now considering only hubs and non-hubs based networks, $\mathcal{C}_t^{\text{hub}}$ and $\mathcal{C}_t^{\text{non-hub}}$, respectively. We hypothesize that the former leads to greater predictability since reflecting information from the industries detected to be the main uncertainty contributors within the system. The uncertainty hubs network is based on the CD, CM, E, IN and IT industries, while the non-hubs network on F, M, RE and U industries, as illustrated in section \ref{totnetdynamic} and in Figure \ref{hubsplot}. Similarly to equation \ref{PredictregresstotControl}, we estimate the following:
\begin{equation} \label{Predicthubsmulti}
\mathcal{Y}^{(\ell)}_{t+h} = \beta_0 + \beta_{\text{hub}} \ \mathcal{C}_t^{\text{hub}} + \beta_{\text{non-hub}} \ \mathcal{C}_t^{\text{non-hub}}  + \sum_{i=1}^{N} \beta_{X,i} \ X_{t,i} + \epsilon_t
\end{equation}
where we add the independent variables that characterize uncertainty hubs and non-hubs based on network connectedness taken jointly and aggregated at a monthly frequency to match the frequency of the indicators we adopt. We include the same set of controls $X$.


In Table \ref{CFNAIHubsnoHubs} we observe that the predictability of the hubs network is superior compared to the non-hubs for the aggregate CFNAI-MA3 and recessions especially for longer horizons, while for expansions at any horizons. We notice how the result achieved by looking at the predictive ability of uncertainty network resemble, or appear even stronger than, the results obtained in the previous section when looking at the aggregated predictability. In Table \ref{CFNAI-MA3Hubs} in the appendix we further confirm this finding by showing the results of the $\mathcal{C}_t^{\text{hub}}$ predictor alone, controlling for $X$. The predictive ability of $\mathcal{C}_t^{\text{hub}}$ is found to be stronger than the one achieved by the aggregate network $\mathcal{C}$. This finding is reflected by the higher adjusted-$R^2$ detected, most of the time, in Table \ref{CFNAI-MA3Hubs} compared to Table \ref{CFNAI-MA3Pre}. These results suggest that the predictive ability of the uncertainty network appear to be driven by a few uncertainty hubs.

\begin{table}[h!]
 \centering
 \caption{Hubs vs. non-Hubs Network Predictive Results} \label{CFNAIHubsnoHubs}
 \begin{threeparttable}
 \centering
\scriptsize{
\begin{tabular}{clcccccc}
 \toprule
&&  \multicolumn{5}{c}{Panel A: CFNAI-MA3}\\
\toprule
   &&  \multicolumn{1}{c}{\textit{h=1}} &  \multicolumn{1}{c}{\textit{h=3}} & \multicolumn{1}{c}{\textit{h=6}} & \multicolumn{1}{c}{\textit{h=9}} & \multicolumn{1}{c}{\textit{h=12}}\\
 \toprule
 $\mathcal{C}_t^{\text{hub}}$ $\mid$ $X_t$ & & -0.007* & -0.013*** & -0.017*** & -0.021*** & -0.019*** \\
  & &(0.004) & (0.004) & (0.004) & (0.004) & (0.004) \\
   $\mathcal{C}_t^{\text{non-hub}}$ $\mid$ $X_t$   &&  -0.015*** & -0.021*** & -0.013* & -0.010 & 0.002 \\
  & &(0.006) & (0.006) & (0.007) & (0.007) & (0.007) \\
 Adj. $R^{2}$ &  &  \multicolumn{1}{c}{0.319} & \multicolumn{1}{c}{0.214} & \multicolumn{1}{c}{0.115} & \multicolumn{1}{c}{0.183} & \multicolumn{1}{c}{0.183} \\
 Obs & & \multicolumn{1}{c}{243} & \multicolumn{1}{c}{241} & \multicolumn{1}{c}{238} & \multicolumn{1}{c}{235} & \multicolumn{1}{c}{232} \\
  \toprule
&&  \multicolumn{5}{c}{Panel B: CFNAI-MA3 Expansion}\\
\toprule
   &&  \multicolumn{1}{c}{\textit{h=1}} &  \multicolumn{1}{c}{\textit{h=3}} & \multicolumn{1}{c}{\textit{h=6}} & \multicolumn{1}{c}{\textit{h=9}} & \multicolumn{1}{c}{\textit{h=12}}\\
 \toprule
 $\mathcal{C}_t^{\text{hub}}$ $\mid$ $X_t$ & & -0.004*** & -0.005*** & -0.005*** & -0.006*** & -0.006*** \\
  & & (0.001) & (0.001) & (0.001) & (0.001) & (0.001) \\
  $\mathcal{C}_t^{\text{non-hub}}$ $\mid$ $X_t$ & & 0.001 & -0.001 & -0.005*** & -0.001 & 0.0004 \\
  & & (0.002) & (0.002) & (0.002) & (0.002) & (0.002) \\
 Adj. $R^{2}$ &  &  \multicolumn{1}{c}{0.097} & \multicolumn{1}{c}{0.150} & \multicolumn{1}{c}{0.258} & \multicolumn{1}{c}{0.289} & \multicolumn{1}{c}{0.286} \\
 Obs & & \multicolumn{1}{c}{243} & \multicolumn{1}{c}{241} & \multicolumn{1}{c}{238} & \multicolumn{1}{c}{235} & \multicolumn{1}{c}{232} \\
 \toprule
&&  \multicolumn{5}{c}{Panel C: CFNAI-MA3 Recession}\\
\toprule
   &&  \multicolumn{1}{c}{\textit{h=1}} &  \multicolumn{1}{c}{\textit{h=3}} & \multicolumn{1}{c}{\textit{h=6}} & \multicolumn{1}{c}{\textit{h=9}} & \multicolumn{1}{c}{\textit{h=12}}\\
 \toprule
 $\mathcal{C}_t^{\text{hub}}$ $\mid$ $X_t$ & & -0.003 & -0.008** & -0.012*** & -0.015*** & -0.012*** \\
  & & (0.004) & (0.004)  &  (0.004) & (0.004) & (0.004) \\
  $\mathcal{C}_t^{\text{non-hub}}$ $\mid$ $X_t$ & & -0.016*** & -0.020*** & -0.009 & -0.009 & 0.002 \\
  &  & (0.005) & (0.006) & (0.007) & (0.006) & (0.006) \\
 Adj. $R^{2}$ &  &   \multicolumn{1}{c}{0.313} & \multicolumn{1}{c}{0.190} & \multicolumn{1}{c}{0.090} & \multicolumn{1}{c}{0.152} & \multicolumn{1}{c}{0.149} \\
 Obs & & \multicolumn{1}{c}{243} & \multicolumn{1}{c}{241} & \multicolumn{1}{c}{238} & \multicolumn{1}{c}{235} & \multicolumn{1}{c}{232} \\
 \bottomrule
\end{tabular}
\begin{tablenotes}
\item {\scriptsize \textit{Notes}: This table presents the results of the predictive regression \ref{Predicthubsmulti} comparing the predictive ability of the uncertainty hubs vs non-hubs sub-networks with respect to the 3-month moving average of the Chicago FED National Activity Index (CFNAI-MA3) in Panel A. In Panel B and Panel C the results of the predictive regression with respect to the CFNAI expansion and recession periods are reported, respectively. The five columns of the table represent different predictability horizons with $h \in (1,3,6,9,12)$. Regressions' coefficients and standard errors (in parentheses), and adjusted-$R^2$ are reported. Coefficients are marked with *, **, *** for 10\%, 5\%, 1\% significance levels, respectively. Intercept and controls results are not reported for the sake of space. Series are considered at a monthly frequency between 01-2000 and 05-2020.    }
\end{tablenotes}
}
\end{threeparttable}
\end{table}

We then check the relationship between the hubs and non-hubs networks for leading indicators. The predictive results for CLI are reported in Table \ref{CLIHubsnoHubs}. We find that the hub network connectedness strongly predicts CLI up to one year, whereas the predictive power of the non-hubs network is overall absent. Thus, the predictive power of the uncertainty hubs network is found to be strong also for business cycles leading indicators, confirming a clear superior predictive ability compared to non-hubs.

We further validate the predictive ability of the hubs-based uncertainty network by including the leading indicator CLI as a control variable in the multivariate regression when predicting CFNAI-MA3. The hubs network shows strong predictive ability from 3-month up to one year, complementing the shorter horizon predictive ability of CLI by expanding it to longer horizons. The non-hub network shows good predictive power in the short horizon and it, therefore, appears not to contain additional information compared to other leading indicators of business cycles e.g. CLI. The hubs-based network shows a longer horizon predictive power, a useful feature for any business cycle leading indicators. Our results suggest that the hubs-based network may be considered as the main driver of the aggregate network, achieving even stronger predictive power on its own.

\begin{table}[h!]
 \centering
 \caption{Hubs vs. non-Hubs Network Predictive Results} \label{CLIHubsnoHubs}
 \begin{threeparttable}
 \centering
\scriptsize{
\begin{tabular}{clcccccc}
 \toprule
&&  \multicolumn{5}{c}{Panel A: Leading Indicator CLI}\\
\toprule
   &&  \multicolumn{1}{c}{\textit{h=1}} &  \multicolumn{1}{c}{\textit{h=3}} & \multicolumn{1}{c}{\textit{h=6}} & \multicolumn{1}{c}{\textit{h=9}} & \multicolumn{1}{c}{\textit{h=12}}\\
\toprule
 $\mathcal{C}_t^{\text{hub}}$ $\mid$ $X_t$ & & -0.012** & -0.024*** & -0.039*** & -0.051*** & -0.056*** \\
  & &(0.005) & (0.006) & (0.006) & (0.006) & (0.006) \\
  $\mathcal{C}_t^{\text{non-hub}}$ $\mid$ $X_t$   &&  -0.009 & -0.015 & -0.016 & -0.017* & -0.005 \\
  & &(0.008) & (0.009) & (0.010) & (0.010) & (0.010) \\
 Adj. $R^{2}$ &  & \multicolumn{1}{c}{0.447} & \multicolumn{1}{c}{0.335} & \multicolumn{1}{c}{0.269} & \multicolumn{1}{c}{0.325} & \multicolumn{1}{c}{0.339} \\
 Obs & & \multicolumn{1}{c}{243} & \multicolumn{1}{c}{241} & \multicolumn{1}{c}{238} & \multicolumn{1}{c}{235} & \multicolumn{1}{c}{232} \\
 \toprule
&&  \multicolumn{5}{c}{Panel B: CFNAI-MA3 controlling for CLI}\\
\toprule
   &&  \multicolumn{1}{c}{\textit{h=1}} &  \multicolumn{1}{c}{\textit{h=3}} & \multicolumn{1}{c}{\textit{h=6}} & \multicolumn{1}{c}{\textit{h=9}} & \multicolumn{1}{c}{\textit{h=12}}\\
 \toprule
 $\mathcal{C}_t^{\text{hub}}$ $\mid$ $X_t$ & & -0.003 & -0.009*** & -0.015*** & -0.020*** & -0.019*** \\
  & & (0.003) & (0.004) & (0.004) & (0.004) & (0.004) \\
  $\mathcal{C}_t^{\text{non-hub}}$ $\mid$ $X_t$   &&  -0.014*** & -0.021*** & -0.012* & -0.011 & 0.002 \\
  & &(0.004) & (0.005) & (0.007) & (0.007) & (0.007) \\
 CLI &&  0.424*** & 0.340*** & 0.243*** & 0.109** & -0.017 \\
  & &(0.036) & (0.044) & (0.051) & (0.051) & (0.053) \\
 Adj. $R^{2}$ &  & \multicolumn{1}{c}{0.573} & \multicolumn{1}{c}{0.371} & \multicolumn{1}{c}{0.193} & \multicolumn{1}{c}{0.196} & \multicolumn{1}{c}{0.179} \\
 Obs & & \multicolumn{1}{c}{243} & \multicolumn{1}{c}{241} & \multicolumn{1}{c}{238} & \multicolumn{1}{c}{235} & \multicolumn{1}{c}{232} \\
\bottomrule
\end{tabular}
\begin{tablenotes}
\item {\scriptsize \textit{Notes}: This table presents the results of the predictive regression \ref{Predicthubsmulti} comparing the predictive ability of the uncertainty hubs vs non-hubs sub-networks with respect to the leading indicator, CLI (Panel A). In Panel B, the results with respect to the CFNAI-MA3 coincident indicator controlling also for the leading indicator, CLI, are reported. The five columns of the table represent different predictability horizons with $h \in (1,3,6,9,12)$. Regressions' coefficients and standard errors (in parentheses), and adjusted-$R^2$ are reported. Coefficients are marked with *, **, *** for 10\%, 5\%, 1\% significance levels, respectively. Intercept and controls results are not reported for the sake of space, the only exception being the CLI control. Series are considered at a monthly frequency between 01-2000 and 05-2020.    }
\end{tablenotes}
}
\end{threeparttable}
\end{table}

As a further robustness check, we also repeat the same predictive exercises by adopting a stricter construction of the hubs and non-hubs networks, including only CM, IN and IT industries and M, RE and U in the networks respectively. We report the results in Table \ref{CFNAIHubsnoHubsSTRICT} in the appendix for all the coincident indicator CFNAI-MA3, expansions and recessions, and also for the leading indicator CLI. We corroborate our previous findings and confirm our hypothesis since the hubs network is found to be more informative in predicting well the future business cycle indicators.

\subsection{Predicting the volatility of GDP}  \label{networkGDP}

Finally, inspired by \cite{carvalho2013}, in this section we check whether or not our uncertainty network connectedness is also able to predict future U.S. GDP growth rate and volatility. We calculate the growth rate of $GDP_t$, the U.S. GDP at time $t$ as \math g_{t} = log(GDP_{t+1}/GDP_{t}) \quad \endmath where $t$ is expressed in quarterly frequency, end of the quarter.
The volatility of GDP growth is measured as the annualized GDP standard deviation over 4 quarters. We check whether or not the aggregate network can predict U.S. GDP indicators in the next $h$ quarters ahead, with $h \in (1,2,3,4)$ by running the predictive equation \ref{PredictregresstotControl} at a quarterly frequency.

We report the empirical results in Table \ref{GDPGRPre} in the appendix. We observe that the uncertainty network cannot predict the future GDP growth rate 2 and 3 quarters ahead. An intensification of connections leads to a decreasing GDP growth rate in the following quarters. We then show how the information embedded in our measures is also useful to predict future GDP volatility. By repeating the same exercise, we also show that the network predicts future GDP volatility in the next four quarters, the results being even stronger and found to be significant up to one year. An increase in connectedness leads to an increase in GDP volatility, thus confirming the counter-cyclicality of the uncertainty network.

In Table \ref{GDPGRPreHUBS} in the appendix, we show the predictive results of hubs and non-hubs based uncertainty networks for the GDP growth rate and volatility. Also, in this case, we find a stronger predictive ability for the hubs network compared to non-hubs for both the GDP growth rate and GDP volatility. We confirm the asymmetric predictive ability of the uncertainty network in favour of the hubs network, while weak and almost absent predictive power for the non-hubs network. We also confirm stronger results for the hubs network compared to the aggregate network results of the previous subsection, emphasizing once more how a network extracted solely from uncertainty hubs might have even strong predictive power not only for the business cycle indicators, but also for the volatility of GDP.

\section{Conclusion}  \label{conclusion}

We studied the ex-ante uncertainty network of the U.S. industries constructed from options-based investors future expectations about the future month uncertainty. We relied on a novel data set of industry forward-looking uncertainties and we adopted a time-varying parameter VAR (TVP-VAR) to model the ex-ante uncertainty network of industries.

We were able to obtain a precise point in time estimation of the uncertainty network to accurately characterize the specific industry role in shocks to uncertainty, dynamically over the business cycle.
We uncovered a main role for booming industries such as communications and information technology and we classified these as uncertainty hubs. Industries such as, financial (having an important influence mainly limited to the global financial crisis), real estate, materials and utilities showed a more neutral role and are classified as uncertainty non-hubs.

We exploited the forward-looking industry connectedness networks characteristics in predictability. We found the industry uncertainty network to be a useful tool to predict future business cycles. We identified a greater predictive ability for the network extracted from uncertainty hubs, being the main (leading) indicator of business cycles. Such uncertainty networks can serve as a new tool for regulators and policymakers to monitor the relationship between industry networks, the business cycle and the real economy in a precise, timely and forward-looking manner.

Fluctuations and shocks to uncertainty in uncertainty hubs should be more carefully monitored due to their potential for shaping the industry networks and impacting the real economy. Our findings suggest a possible direction for policies and government interventions. Policy actions should target uncertainty hubs since they are the stronger contributors to uncertainty shocks and they show a tighter link with the real economy.
Policy interventions aimed at dampening their shocks to uncertainty can potentially provide a direct channel for boosting business cycles.



\clearpage
\newpage
\singlespacing
\bibliography{Industry_Vol_Connect}
\bibliographystyle{chicago}

\newpage

\begin{appendices}

\section*{Appendix}
\label{app:sec00}

\setcounter{table}{0}
\setcounter{table}{0}
\renewcommand{\thetable}{A\arabic{table}}

\renewcommand{\thefigure}{A\arabic{figure}}

\setcounter{figure}{0}

\section{S\&P 500 sectors breakdown} \label{breakdownindustries}

In this short section we present a U.S. stock market sectors breakdown and description where the S\&P 500 index is used as a proxy for the stock market. The information in this section are reported as of January 25, 2019. For more details and updated information, see also \url{https://us.spindices.com/indices/equity/sp-500}.

\begin{itemize}
  \item Consumer Discretionary (CD): The CD sector consists of businesses that have demand that rises and falls based on general economic conditions such as washers and dryers, sporting goods, new cars, etc. At present, the consumer discretionary sector contains 11 sub-industries: Automobile Components Industry, Automobiles Industry, Distributors Industry, Diversified Consumer Services Industry, Hotels, Restaurants \& Leisure Industry, Household Durables Industry, Leisure Products Industry, Multiline Retail Industry, Specialty Retail Industry, Textile, Apparel \& Luxury Goods Industry, Internet \& Direct Marketing. The total value of all consumer discretionary stocks in the U.S. came to \$4.54 trillion, or about 10.11\% of the market. Examples of consumer discretionary stocks include Amazon and Starbucks.
  \item Communication Services (CM): from telephone access to high-speed internet, the communication services sector of the economy keeps us all connected. At present, the communication services sector is made up of five industries: Diversified Telecommunication Services, Wireless Telecommunication Services, Entertainment Media, Interactive Media and Services. the total value of all communication services stocks in the U.S. came to \$4.42 trillion, or 10.33\% of the market. The CM industry includes stocks such as AT\&T and Verizon, but also the giants Alphabet Inc A and Facebook from 2004 and 2012, respectively. In fact, the new communication sector of the S\&P 500 includes now big companies such as Facebook and Alphabet Google since these were moved out from the technology and consumer discretionary sectors, respectively, due to the changes of the Global Industry Classification Standard (GICS).
  \item Consumer Staples (CS): The CS sector consists of businesses that sell the necessities of life, ranging from bleach and laundry detergent to toothpaste and packaged food. At present, the consumer staples sector contains six industries:  Beverages Industry, Food \& Staples Retailing Industry, Food Products Industry, Household Products Industry, Personal Products Industry, Tobacco Industry. The total value of all consumer staples stocks in the U.S. came to \$2.95 trillion, or about 7.18\% of the market and includes companies such as Procter \& Gamble.
  \item Energy (E): The E sector consists of businesses that source, drill, extract, and refine the raw commodities we need to keep the country going, such as oil and gas. At present, the energy sector contains two industries: Energy Equipment \& Services Industry, and Oil, Gas \& Consumable Fuels Industry. The total value of all energy stocks in the U.S. came to \$3.36 trillion, or about 5.51\% of the market. Major energy stocks include Exxon Mobil and Chevron.
  \item Financial (F): The F sector consists of banks, insurance companies, real estate investment trusts, credit card issuers. At present, the financial sector contains seven industries: Banking Industry, Capital Markets Industry, Diversified Financial Services Industry, Insurance Industry, Mortgage Real Estate Investment Trusts (REITs) Industry, Thrifts \& Mortgage Finance Industry. The total value of all financial stocks in the U.S. came to \$6.89 trillion, or about 13.63\% of the market. J.P. MorganChase, GoldmanSachs, and Bank of America are examples of financial stocks.
  \item Health Care (HC): The HC sector consists of drug companies, medical supply companies, and other scientific-based operations that are concerned with improving and healing human life. At present, the HC sector contains six industries: Biotechnology Industry, Health Care Equipment \& Supplies Industry, Health Care Providers \& Services Industry, Health Care Technology Industry, Life Sciences Tools \& Services Industry, Pharmaceuticals Industry. The total value of all health care stocks in the U.S. came to \$5.25 trillion, or about 15.21\% of the market. Examples of HC stocks include Johnson \& Johnson, and Pfizer.
  \item Industrials (IN): The IN sector comprises railroads and airlines to military weapons and industrial conglomerates. At present, the industrial sector contains fourteen industries: Aerospace \& Defense Industry, Air Freight \& Logistics Industry, Airlines Industry, Building Products Industry, Commercial Services \& Supplies Industry, Construction \& Engineering Industry, Electrical Equipment Industry, Industrial Conglomerates Industry, Machinery Industry, Marine Industry, Professional Services Industry, Road \& Rail Industry, Trading Companies \& Distributors Industry, Transportation Infrastructure Industry. The total value of all IN stocks in the U.S. came to \$3.80 trillion, or about 9.33\% of the market.
  \item Information Technology (IT): the IT sector is home to the hardware, software, computer equipment, and IT services operations that make it possible for you to be reading this right now. At present, the information technology sector contains six industries: Communications Equipment Industry, Electronic Equipment, Instruments \& Components Industry, IT Services Industry, Semiconductors \& Semiconductor Equipment Industry, Software Industry, Technology Hardware, Storage \& Peripherals Industry. The total value of all IT stocks in the United States came to \$7.10 trillion, or about 19.85\% of the market. It is the largest sector in the S\&P 500. Top IT stocks include Microsoft and Apple.
    \item Materials (M): The building blocks that supply the other sectors with the raw materials it needs to conduct business, the material sector manufacturers, logs, and mines everything from precious metals, paper, and chemicals to shipping containers, wood pulp, and industrial ore. At present, the material sector contains five industries: Chemicals Industry, Construction Materials Industry, Containers \& Packaging Industry, Metals \& Mining Industry, Paper \& Forest Products Industry. The total value of all materials stocks in the U.S. came to \$1.77 trillion, or about 2.71\% of the market. Major materials stocks include Dupont.
  \item Real Estate (RE): The RE sector includes all Real Estate Investment Trusts (REITs) with the exception of Mortgage REITs, which is housed under the financial sector. The sector also includes companies that manage and develop properties. At present, the RE sector is made up of two industries: Equity Real Estate Investment Trusts, Real Estate Management \& Development. The total value of all real estate stocks in the U.S. came to \$1.17 trillion, or 2.96\% of the market. The RE industry includes stocks such as Simon Property Group and Prologis.
  \item Utilities (U): The U sector of the economy is home to the firms that make our lights work when we flip the switch, let our stoves erupt in flame when we want to cook food, make water come out of the tap when we are thirsty, and more. At present, the utilities sector is made up of five industries: Electric Utilities Industry, Gas Utilities Industry, Independent Power and Renewable Electricity Producers Industry, Multi-Utilities Industry, Water Utilities Industry. The total value of all utilities stocks in the U.S. came to \$1.27 trillion, or about 3.18\% of the market. U stocks include many local electricity and water companies including Dominion Resources.
\end{itemize}

\newpage

\section{Model-free individual implied volatility}
\label{sec: appB}

Formalizing the implied volatility computation for each stock, we follow \cite{Bakshi2003} in adopting out-of-money (OTM) call and put option prices to compute the individual stock $s$th implied variance as
\begin{equation}
\label{eq:ivdecomp}
\sigma^2_{\text{VIX}^{(s)}} = \int_{P_t}^{\infty} \frac{2(1-\log(K/P_t))}{K^2}C(t,t+1,K) dK  + \int_{0}^{P_t} \frac{2(1+\log(P_t/K))}{K^2}P(t,t+1,K) d K,
\end{equation}
where $C(.)$ and $P(.)$ denote the time $t$ prices of call and put contracts, respectively, with time to maturity of one period and a strike price of $K$.

Intuitively, the implied variance measure can be computed in a model-free way from a range of option prices upon a discretization of formula (\ref{eq:ivdecomp}), adopting call and put option prices with respect to the next 30 days, considering all available strikes for each individual stock options. We compute the $\text{VIX}^{(s)}$ for all the stocks in our sample belonging to the 11 U.S. industries as follows:
\begin{equation}\label{eq:vix}
\sigma^2_{\text{VIX}^{(s)}}=\frac{2}{T}\sum_{i=1}^{n}\frac{\Delta K_i}{K_i^2}e^{rT}Q(K_i)-\frac{1}{T}\left[\frac{F}{K_0}-1\right]^2,\\
\end{equation}
where $T$ is time to expiration, $F$ is the forward index level derived from the put-call parity as $F= e^{rT}[C(K,T)-P(K,T)]+K$ with the risk-free rate $r$
, $K_0$ is the reference price, the first exercise price less or equal to the forward level $F(K_0 \leq F)$, and $K_i$ is the $i$th out-of-the-money (OTM) strike price available on a specific date (call if $K_i>K_0$, put if $K_i<K_0$, and both call and put if $K_i = K_0$). $Q(K_i)$ is the average bid-ask of OTM options with exercise price equal to $K_i$. If $K_i=K_0$, it will be equal to the average between the at-the-money (ATM) call and put price, relative to the strike price, and $\Delta(K_i)$ is the sum divided by two of the two nearest prices to the exercise price $K_0$, namely, $\frac{(K_{i+1}-K_{i-1})}{2}$ for $2\leq i \leq n-1$.
The annualized square roots of the quantities computed for each of the $s$-th individual firms are then labeled $\text{VIX}^{(s)}$ denoting individual, model-free implied volatility measures of the expected price fluctuations in the $s$-th underlying asset's options over the next month.\footnote{The standard CBOE methodology considers an interpolation between the two closest to 30-days expiration dates. We use a simplified formula taking into account only one expiration date closest to 30-days due to options data availability with respect to U.S. single stocks.}


\begin{equation}   \label{vixeq}
\text{VIX}^{(s)} =  \sqrt{\frac{365}{30} \sigma^{2}_{\text{VIX}^{(s)}}}
\end{equation}

\clearpage
\newpage

\setcounter{table}{0}
\setcounter{table}{0}
\renewcommand{\thetable}{B\arabic{table}}

\renewcommand{\thefigure}{B\arabic{figure}}

\setcounter{figure}{0}

\begin{landscape}
\begin{table}[h]
 \centering
 \caption{\textbf{List of Selected Stocks by Industry}} \label{StocksInfo}
 \begin{threeparttable}
 \centering
 \tiny{
 \begin{tabular}{ccccccc}
 \toprule
Ticker	&	Full name	&	Period	&&	Ticker	&	Full name	&	Period	\\
 \cmidrule{1-3} \cmidrule{5-7}
Consumer discretionary: 	&		&		&&	Industrial	&		&		\\
 \cmidrule{1-3} \cmidrule{5-7}
AMZN	&	Amazon.com Inc.	&	1996-2020	&&	BA	&	Boeing Company	&	1996-2020	\\
HD	&	Home Depot Inc.	&	1996-2020	&&	GE	&	General Electric Company	&	1996-2020	\\
MCD	&	McDonald's Corporation	&	1996-2020	&&	HON	&	Honeywell International Inc.	&	1996-2020	\\
NKE	&	NIKE Inc. Class B	&	1996-2020	&&	UNP	&	Union Pacific Corporation	&	1996-2020	\\
SBUX	& Starbucks Corporation		&	1996-2020	&&	UTX	&	United Technologies Corporation	&	1996-2018	\\
& &                                                 && CAT & Caterpillar Inc. & 2019-2020\\
\cmidrule{1-3} \cmidrule{5-7}
Communications: 	&		&		&&	Information Technology	&		&		\\
 \cmidrule{1-3} \cmidrule{5-7}
CMCSA	& Comcast Corporation Class A		&	1996-2004	&&	AAPL	& Apple Inc.		&	1996-2020	\\
DIS	&	Walt Disney Company	&	1996-2020	&&	ADBE	&	Adobe Inc.	&	1996-2020	\\
EA	&	Electronic Arts Inc.	&	1996-2015	&&	CSCO	&	Cisco Systems Inc.	&	1996-2015	\\
FB	&	Facebook Inc. Class A	&	2012-2020	&&	INTC	&	Intel Corporation	&	1996-2020	\\
GOOG 	&	Alphabet Inc. Class C	&	2004-2020	&&	MA	&	Mastercard Incorporated Class A	&	2007-2018	\\
T	&	AT\&T Inc.	&	1996-2020	&&	MSFT	&	Microsoft Corporation 	&	1996-2020	\\
VZ	&	Verizon Communications Inc.	&	1996-2020	&&	V	&	Visa Inc. Class A	&	2008-2020	\\
\cmidrule{1-3} \cmidrule{5-7}
Consumer staples: 	&		&&		&	Materials	&		&		\\
\cmidrule{1-3} \cmidrule{5-7}
COST	&Costco Wholesale Corporation		&	1996-2018	&&	APD	&	Air Products and Chemicals Inc	&	1996-2018	\\
KO	&	Coca-Cola Company	&	1996-2020	&&	DD	&	DuPont de Nemours Inc.	&	1996-2018	\\
PEP	&	PepsiCo Inc.	&	1996-2020	&&	ECL	&	Ecolab Inc.	&	1996-2018	\\
PG	&	Procter \& Gamble Company	&	1996-2020	&&	LIN	&	Linde plc	&	1996-2007	\\
WMT	&	Walmart Inc.	&	1996-2020	&&	PPG	&	PPG Industries Inc.	&	2007-2018	\\
PM   & Philip Morris Inc.&    2018-2020                 &&  SHW	&	Sherwin-Williams Company	&	1998-2020	\\
 & &                                                    && NEM & Newmont Corporation & 2019-2020\\
 & &                                                     && FCX & Freeport-McMoRan Inc. & 2019-2020\\
 & &                                                       && IP & 	International Paper & 2019-2020\\
\cmidrule{1-3}  \cmidrule{5-7}
Energy: 	&		&		&&	Real Estate	&		&		\\
\cmidrule{1-3}  \cmidrule{5-7}
COP	&	ConocoPhillips	&	1996-2020	&&	AMT	&	American Tower Corporation	&	1996-2020	\\
CVX	&	Chevron Corporation	&	1996-2020	&&	CCI	&	Crown Castle International Corp	&	1996-2018	\\
SLB	&	Schlumberger NV	&	1996-2020	&&	EQIX	&	Equinix Inc.	&	2006-2020	\\
XOM 	& Exxon Mobil Corporation		&	1996-2020	&&	EQR	&	Equity Residential	&	1996-2020	\\
VLO 	&	Valero Energy Corp	&	1996-2015 	&&	PLD	&	Prologis Inc.	&	1996-2018	\\
PSX	&	Phillips 66	&	 2012-2020	&&	SPG	&	Simon Property Group Inc.	&	1996-2020	\\
\cmidrule{1-3}  \cmidrule{5-7}
Financial	&		&		&&	Utilities	&		&		\\
\cmidrule{1-3}  \cmidrule{5-7}
BAC	&	Bank of America Corp	&	1996-2020	&&	AEP	&	American Electric Power Company Inc.	&	1996-2020	\\
BRK	&	Berkshire Hathaway Inc. Class B	&	2010-2020	&&	D	&	Dominion Energy Inc	&	1996-2020	\\
C	&	Citigroup Inc.	&	1996-2020	&&	DUK	&	Duke Energy Corporation	&	2006-2020	\\
GS	&	Goldman Sachs	&	1999-2018	&&	NEE	&	NextEra Energy Inc.	&	1996-2020	\\
JPM	&	JPMorgan Chase \& Co.	&	1996-2020	&&	SO	&	Southern Company	&	1996-2020	\\
WFC	&	Wells Fargo \& Company	&	1996-2020	&&   & & \\
\cmidrule{1-3}
Health care: 	&		&		&&		&		&		\\
\cmidrule{1-3}
ABT	&	Abbott Laboratories	&	1996-2019	&&		&		&		\\
JNJ & Johnson \& Johnson      &   1996-2020   &&      &       &       \\
MRK	&	Merck \& Co. Inc.	&	1996-2020	&&		&		&		\\
PFE	&	Pfizer Inc.	&	1996-2020	&&		&		&		\\
 UNH 	&	UnitedHealth Group Incorporated	&	1996-2020	&&		&		&		\\
 ABBV & AbbVie Inc. & 2018-2020 && &		&		\\
\bottomrule
\end{tabular}
\begin{tablenotes}
\item {\scriptsize \textit{Notes}: This table summarizes all the U.S. stocks selected in the paper, divided by industry. The stock tickers, full names and period of options data availability are reported.   }
\end{tablenotes}
}
\end{threeparttable}
\end{table}
\end{landscape}

\begin{figure}[ht!]
 \begin{center}
\caption{\textbf{Individual Firm Uncertainty}} \label{IVIX}
\includegraphics[width=\textwidth]{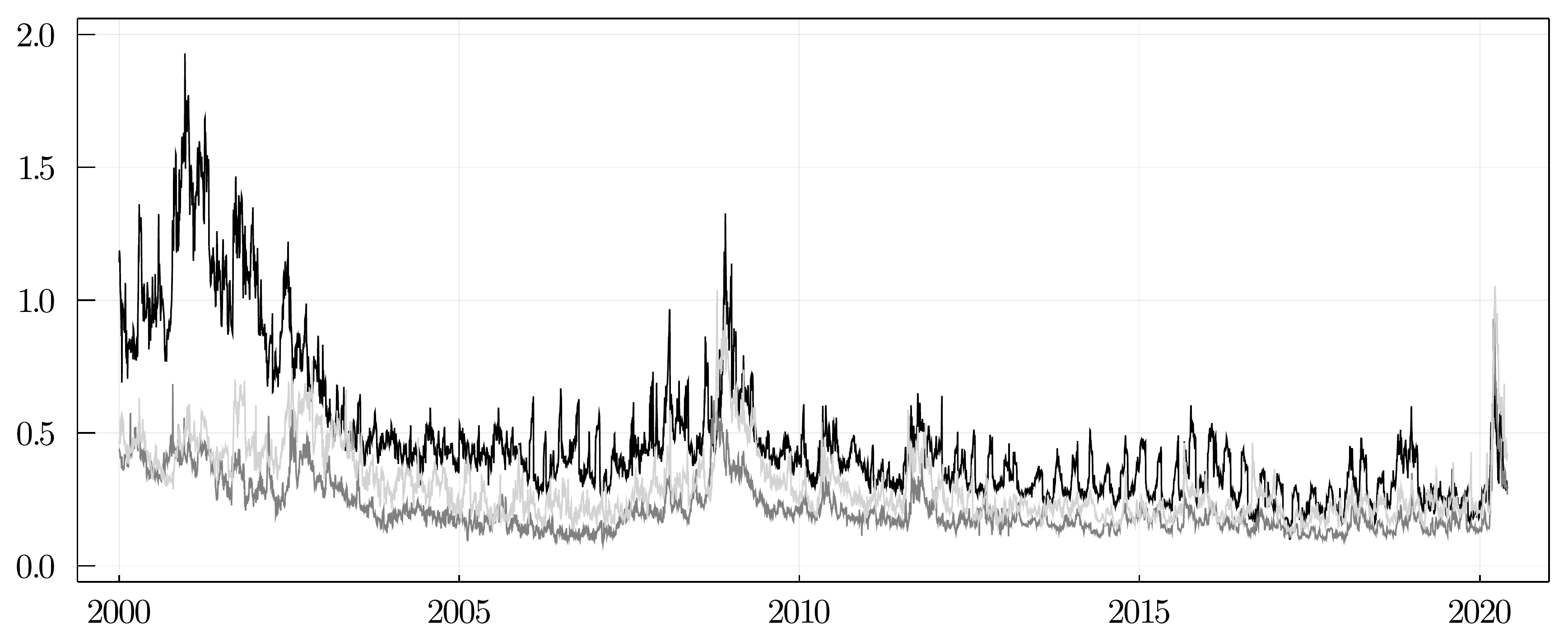}
\caption*{\textit{Notes}: This figure illustrates the individual uncertainty $\text{VIX}^{(i)}$ for Amazon (in black), CocaCola (in grey) and Disney (in light grey) from 03-01-2000 to 29-05-2020, at a daily frequency. }
\end{center}
\end{figure}

\section{Proofs}\label{app:proofs}
\begin{proof}[Proposition~\ref{prop:1}]

Let us have the VMA($\infty$) representation of the locally stationary TVP-VAR model  \citep{dahlhaus2009empirical,roueff2016prediction}
\begin{equation}
\bX_{t,T} = \sum_{h=-\infty}^{\infty} \bPsi_{t,T,h}\bepsilon_{t-h}
\end{equation}
$\bPsi_{t,T,h} \approx\bPsi(t/T,h)$ is a stochastic process satisfying $\sup_{\ell} ||\bPsi_t-\bPsi_{\ell}||^2 = O_p(h/t)$ for $1\le h \le t$ as $t\rightarrow \infty$, hence in a neighborhood of a fixed time point $u=t/T$ the process $\bX_{t,T}$ can be approximated by a stationary process $\widetilde{\bX}_t(u)$
\begin{equation}
\widetilde{\bX}_t(u) = \sum_{h=-\infty}^{\infty} \bPsi_h(u)\bepsilon_{t-h}
\end{equation}
with $\bepsilon$ being \textit{iid} process with $\mathbb{E}[\bepsilon_t]=0$, $\mathbb{E}[\bepsilon_s\bepsilon_t]=0$ for all $s\ne t$, and the local covariance matrix of the errors $\bSigma(u)$. Under suitable regularity conditions $|\bX_{t,T} - \widetilde{\bX}_t(u)| = O_p\big( |t/T-u|+1/T\big)$.

Since the errors are assumed to be serially uncorrelated, the total local covariance matrix of the forecast error conditional on the information at time $t-1$ is given by
\begin{equation}
	\bOmega^H(u) = \sum_{h=0}^{H} \bPsi_h(u) \bSigma(u) \bPsi^{\top}_h(u).
\end{equation}
Next, we consider the local covariance matrix of the forecast error conditional on knowledge of today's shock and future expected shocks to $k$-th variable. Starting from the conditional forecasting error,
\begin{equation}
\label{eq:LGFEVD1}
	\boldsymbol \xi^{k,H}(u) = \sum_{h=0}^{H} \bPsi_h(u) \Big[ \bepsilon_{t+H-h} - \mathbb{E}(\bepsilon_{t+H-h} | \bepsilon_{k,t+H-h})  \Big],
\end{equation} assuming normal distribution of $\bepsilon_t\sim N(0,\bSigma)$, we obtain\footnote{Note to notation: $[\boldsymbol A]_{j,k}$ denotes the $j$th row and $k$th column of matrix $\boldsymbol A$ denoted in bold. $[\boldsymbol A]_{j,\cdot}$ denotes the full $j$th row; this is similar for the columns. A $\sum A$, where $A$ is a matrix that denotes the sum of all elements of the matrix $A$.}
\begin{equation}
\label{eq:LGFEVD2}
\mathbb{E}(\bepsilon_{t+H-h} | \bepsilon_{k,t+H-h}) = \sigma_{kk}^{-1} \Big[\bSigma(u) \Big]_{\cdot k} \bepsilon_{k,t+H-h}
\end{equation}
and substituting (\ref{eq:LGFEVD2}) to (\ref{eq:LGFEVD1}), we obtain
\begin{equation}
	\boldsymbol \xi^{k,H}(u) = \sum_{h=0}^{H} \bPsi_h(u) \Big[ \bepsilon_{t+H-h} - \sigma_{kk}^{-1} \Big[\bSigma(u) \Big]_{\cdot k} \bepsilon_{k,t+H-h}  \Big].
\end{equation}
Finally, the local forecast error covariance matrix is
\begin{equation}
	\bOmega^{k,H}(u) = \sum_{h=0}^{H} \bPsi_h(u) \bSigma(u) \bPsi^{\top}_h(u) - \sigma_{kk}^{-1} \sum_{h=0}^{H} \bPsi_h(u) \Big[\bSigma(u) \Big]_{\cdot k} \Big[\bSigma(u) \Big]_{\cdot k}^{\top} \bPsi^{\top}_h(u).
\end{equation}
Then
\begin{equation}
	\Big[\boldsymbol \Delta^H(u)\Big]_{(j)k} = \Big[\bOmega^H(u) - \bOmega^{k,H}(u)\Big]_{j,j} = \sigma_{kk}^{-1} \sum_{h = 0}^{H} \Bigg( \Big[\bPsi_h(u) \bSigma(u)\Big]_{j,k} \Bigg)^2
\end{equation} is the unscaled local $H$-step ahead forecast error variance of the $j$-th component with respect to the innovation in the $k$-th component. Scaling the equation with $H$-step ahead forecast error variance with respect to the $j$th variable yields the desired time-varying generalized forecast error variance decompositions (TVP-GFEVD)
\begin{equation}
\label{eq:tvpgfevd}
	\Big[ \btheta^H(u) \Big]_{j,k} = \frac{\sigma_{kk}^{-1}\displaystyle\sum_{h = 0}^{H}\Bigg( \Big[\bPsi_h(u) \bSigma(u)\Big]_{j,k} \Bigg)^2}{\displaystyle \sum_{h=0}^{H} \Big[ \bPsi_h(u) \bSigma(u) \bPsi^{\top}_h(u) \Big]_{j,j}}
\end{equation}

	This completes the proof.
\end{proof}

\section{Estimation of the time-varying parameter VAR model} \label{app:estimate}

To estimate our high dimensional systems, we follow the Quasi-Bayesian Local-Liklihood (QBLL) approach of \cite{petrova2019quasi}. let $\bX_{t}$ be an $N \times 1$ vector generated by a stable time-varying parameter (TVP) heteroskedastic VAR model with $p$ lags:

\begin{equation}\label{eq:VAR}
\bX_{t,T}=\bPhi_{1}(t/T)\bX_{t-1,T}+\ldots+\bPhi_{p}(t/T)\bX_{t-p,T} + \bepsilon_{t,T},
\end{equation}
where $\bepsilon_{t,T}=\bSigma^{-1/2}(t/T)\bbeta_{t,T}$ with $\bbeta_{t,T}\sim NID(0,\boldsymbol{I}_M)$ and $\bPhi(t/T)=(\bPhi_{1}(t/T),\ldots,\bPhi_{p}(t/T))^{\top}$ are the time-varying autoregressive coefficients.
Note that all roots of the polynomial, $\chi(z)=\text{det}\left(\textbf{I}_{N}-\sum^{L}_{p=1}z^{p}\mathbf{B}_{p,t}\right)$, lie outside the unit circle, and $\bSigma^{-1}_{t}$ is a positive definite time-varying covariance matrix. Stacking the time-varying intercepts and autoregressive matrices in the vector $\phi_{t,T}$ with $\overline{\bX}^{\top}_{t} = \left(\text{\textbf{I}}_{N} \otimes x_{t}\right),\: x_{t}=\left(1,x^{\top}_{t-1},\dots,x^{\top}_{t-p}\right)$ and $\otimes$ denotes the Kronecker product, the model can be written as:
\begin{eqnarray}
\bX_{t,T} = \overline{\bX}^{\top}_{t,T}\phi_{t,T} + \bSigma^{-\frac{1}{2}}_{t/T}\bbeta_{t,T}
\end{eqnarray}
We obtain the time-varying parameters of the model by employing Quasi-Bayesian Local Likelihood (QBLL) methods. Estimation of (\ref{eq:VAR}) requires re-weighting the likelihood function. Essentially, the weighting function gives higher proportions to observations surrounding the time period whose parameter values are of interest. The local likelihood function at time period $k$ is given by:

\begin{eqnarray}
\begin{split}
\mathbf{L}_{k}\left(\bX|\theta_{k},\bSigma_{k},\overline{\bX} \right)&\propto& \\ |\bSigma_{k}|^{\text{trace}(\mathbf{D}_{k})/2}\exp\left\{-\frac{1}{2}(\bX-\overline{\bX}^{\top}\phi_{k})^{\top}\left(\bSigma_{k}\otimes\mathbf{D}_{k}\right)(\bX-\overline{\bX}^{\top}\phi_{k})\right\}
\end{split}
\end{eqnarray}

The $\mathbf{D}_{k}$ is a diagonal matrix whose elements hold the weights:
\begin{eqnarray}
\mathbf{D}_{k} &=& \text{diag}(\varrho_{k1},\dots,\varrho_{kT})\\
\varrho_{kt} &=& \phi_{T,k}w_{kt}/\sum^{T}_{t=1}w_{kt}\\
\label{eq:weight}
w_{kt} &=& (1/\sqrt{2\pi})\exp((-1/2)((k-t)/H)^{2}),\quad\text{for}\: k,t\in\{1,\dots,T\}\\
\zeta_{Tk} &=& \left(\left(\sum^{T}_{t=1}w_{kt}\right)^{2}\right)^{-1}
\end{eqnarray}
where $\varrho_{kt}$ is a normalised kernel function. $w_{kt}$ uses a Normal kernel weighting function. $\zeta_{Tk}$ gives the rate of convergence and behaves like the bandwidth parameter $H$ in (\ref{eq:weight}), and it is the kernel function that provides greater weight to observations surrounding the parameter estimates at time $k$ relative to more distant observations.

Using a Normal-Wishart prior distribution for $\phi_{k}|\:\bSigma_{k}$ for $k\in\{1,\dots,T\}$:
\begin{eqnarray}
\phi_{k}|\bSigma_{k} \backsim \mathcal{N}\left(\phi_{0k},(\bSigma_{k} \otimes \mathbf{\Xi}_{0k})^{-1}\right)\\
\bSigma_{k} \backsim \mathcal{W}\left(\alpha_{0k},\mathbf{\Gamma}_{0k}\right)
\end{eqnarray}
where $\phi_{0k}$ is a vector of prior means, $\mathbf{\Xi}_{0k}$ is a positive definite matrix, $\alpha_{0k}$ is a scale parameter of the Wishart distribution ($\mathcal{W}$), and $\mathbf{\Gamma}_{0k}$ is a positive definite matrix.

The prior and weighted likelihood function implies a Normal-Wishart quasi posterior distribution for $\phi_{k}|\:\bSigma_{k}$ for $k=\{1,\dots,T\}$. Formally let $\mathbf{A} = (\overline{x}^{\top}_{1},\dots,\overline{x}^{\top}_{T})^{\top}$ and $\mathbf{Y}=(x_{1},\dots,x_{T})^{\top}$ then:
\begin{eqnarray}
\phi_{k}|\bSigma_{k},\mathbf{A},\mathbf{Y} &\backsim & \mathcal{N}\left(\tilde{\theta}_{k},\left(\bSigma_{k}\otimes\mathbf{\tilde{\Xi}}_{k}\right)^{-1}\right)\\
\bSigma_{k} &\backsim & \mathcal{W}\left(\tilde{\alpha}_{k},\mathbf{\tilde{\Gamma}}^{-1}_{k} \right)
\end{eqnarray}
with quasi posterior parameters
\begin{eqnarray}
\tilde{\phi}_{k} &=& \left(\mathbf{I}_{N}\otimes \mathbf{\tilde{\Xi}}^{-1}_{k}\right)\left[\left(\mathbf{I}_{N}\otimes \mathbf{A}'\mathbf{D}_{k}\mathbf{A}\right)\hat{\phi}_{k}+ \left(\mathbf{I}_{N}\otimes \mathbf{\Xi}_{0k}\right)\phi_{0k} \right]\\
\mathbf{\tilde{\Xi}}_{k} &=& \mathbf{\tilde{\Xi}}_{0k} + \mathbf{A}'\mathbf{D}_{k}\mathbf{A}\\
\tilde{\alpha}_{k} &=& \alpha_{0k}+\sum^{T}_{t=1}\varrho_{kt}\\
\mathbf{\tilde{\Gamma}}_{k} &=& \mathbf{\Gamma}_{0k} + \mathbf{Y}'\mathbf{D}_{k}\mathbf{Y} + \mathbf{\Phi}_{0k}\mathbf{\Gamma}_{0k}\mathbf{\Phi}'_{0k} - \mathbf{\tilde{\Phi}}_{k}\mathbf{\tilde{\Gamma}}_{k}\mathbf{\tilde{\Phi}}^{\top}_{k}
\end{eqnarray}
where $\hat{\phi}_{k} = \left(\mathbf{I}_{N}\otimes \mathbf{A}'\mathbf{D}_{k}\mathbf{A}\right)^{-1}\left(\mathbf{I}_{N} \otimes \mathbf{A}'\mathbf{D}_{k}\right)y$ is the local likelihood estimator for $\phi_{k}$. The matrices $\mathbf{\Phi}_{0k},\:\mathbf{\tilde{\Phi}}_{k}$ are conformable matrices from the vector of prior means, $\phi_{0k}$, and a draw from the quasi posterior distribution, $\tilde{\phi}_{k}$, respectively.

The motivation for employing these methods are threefold. First, we are able to estimate large systems that conventional Bayesian estimation methods do not permit. This is typically because the state-space representation of an $N$-dimensional TVP-VAR ($p$) requires an additional $N(3/2 + N(p+1/2))$ state equations for every additional variable. Conventional Markov Chain Monte Carlo (MCMC) methods fail to estimate larger models, which in general confine one to (usually) fewer than 6 variables in the system. Second, the standard approach is fully parametric and requires a law of motion. This can distort inference if the true law of motion is misspecified. Third, the methods used here permit direct estimation of the VAR's time-varying covariance matrix, which has an inverse-Wishart density and is symmetric positive definite at every point in time.

In estimating the model, we use $p$=2 and a Minnesota Normal-Wishart prior with a shrinkage value $\varphi=0.05$ and centre the coefficient on the first lag of each variable to 0.1 in each respective equation. The prior for the Wishart parameters are set following \cite{kadiyala1997numerical}. For each point in time, we run 500 simulations of the model to generate the (quasi) posterior distribution of parameter estimates. Note we experiment with various lag lengths, $p=\{2,3,4,5\}$; shrinkage values, $\varphi=\{0.01, 0.25, 0.5\}$; and values to centre the coefficient on the first lag of each variable, $\{0, 0.05, 0.2, 0.5\}$. Network measures from these experiments are qualitatively similar. Notably, adding lags to the VAR  and increasing the persistence in the prior value of the first lagged dependent variable in each equation increases computation time.

\section{Net and Agg Uncertainty Connectedness Measures}

\setcounter{table}{0}
\setcounter{table}{0}
\renewcommand{\thetable}{E\arabic{table}}

\renewcommand{\thefigure}{E\arabic{figure}}

\setcounter{figure}{0}

\begin{table} [ht!]
  \centering
  \caption{\textbf{$\mathcal{C}_{\textsc{NET}}$ and $\mathcal{C}_{\textsc{AGG}}$ across each business cycle}}   \label{CumNetBusinesscylce}
  \begin{threeparttable}
  \scriptsize{
  \centering
  \begin{tabular}{ccccccccccccccccc}
  \toprule
 &&  \multicolumn{3}{c}{\textbf{Dot com Inv}} &&  \multicolumn{3}{c}{\textbf{Dot com Rec}}&&  \multicolumn{3}{c}{\textbf{Exp after Dot com}} &&  \multicolumn{3}{c}{\textbf{GFC Inv}}\\		
\cmidrule{1-5}	 \cmidrule{7-9}	\cmidrule{11-13}	\cmidrule{15-17}	
 &&	  NET	&AGG & AGG $\%$  &&	 NET & AGG & AGG $\%$ &&NET	&AGG & AGG $\%$   && NET	&AGG  & AGG $\%$ \\
 \cmidrule{1-5}	 \cmidrule{7-9}	\cmidrule{11-13}	\cmidrule{15-17}	
CD&&	-2.27&	20.81 & 10    &&	0.81&	66.58 & 13.5 &&	-2.04&	10.99 & 8.2  &&	-1.17&	12.72  & 9.4\\
CM&&	-3.09&	14.28&   6.9  &&	-0.26&	38.82& 7.8   &&	-1.22&	11.73& 8.8   &&	-0.02&	11.85  & 8.8\\
CS&&	-1.02&	13.08& 6.3    &&	-4.91&	47.86& 9.7   &&	-0.76&	11.38& 8.6   &&	-0.72&	11.17& 8.3\\
E&&	    -6.24&	21.21& 10.2   &&	0.37&	42.96& 8.7   &&	-0.71&	12.57& 9.4   &&	-2.79&	12.54& 9.3\\
F&&	     0.14&	18.77& 9.1      &&	0.06&	66.84& 13.5   && 1.51 &	16.86& 12.7  &&	7.38&	23.37& 17.3\\
HC&&	-1.31&	18.31& 8.8    &&	-4.11&	54.12& 11   &&	-0.42&	12.24& 9.2   &&	-0.21&	10.14& 7.5\\
IN&&	 1.79&	28.56& 13.8   &&	-3.27&	62.47& 12.6   &&	0.96&	13.47& 10.1  &&	0.05&	12.70& 9.4\\
IT&&	 9.79&	36.39& 17.6   &&	6.48&	68.59& 13.9   &&	1.03&	12.51& 9.4   &&	-0.94&	13.79& 10.2\\
M&&	    -0.62& 	10.64& 5.1    &&	-1.47&	13.14& 2.7    &&	-1.03&	9.24& 6.9   &&	0.04&	9.66& 7.1\\
RE&&	 3.68&	16.34& 7.9    &&	7.02&	24.92& 5.1    &&	2.52&	13.09& 9.8   &&	-1.99&	7.67& 5.7\\
U&&	     -0.85&	8.71& 4.3       &&	-0.67&	7.89& 1.5    &&	0.17&	9.12& 6.9   &&	-0.32&	9.50& 7\\
\toprule
 &&  \multicolumn{3}{c}{\textbf{GFC Rec}} &&  \multicolumn{3}{c}{\textbf{Exp after GFC}}&&  \multicolumn{3}{c}{\textbf{Covid-19 Rec}} &&  \multicolumn{3}{c}{\textbf{Total Period}}\\		
\cmidrule{1-5}	 \cmidrule{7-9}	\cmidrule{11-13}	\cmidrule{15-17}	
  &&	  NET	&AGG & AGG $\%$  &&	 NET & AGG & AGG $\%$ &&NET	&AGG & AGG $\%$   && NET	&AGG  & AGG $\%$ \\
 \cmidrule{1-5}	 \cmidrule{7-9}	\cmidrule{11-13}	\cmidrule{15-17}	
CD&&	-2.98&	16.12  & 6.5  &&	0.07&	40.24  & 12.1  &&	3.64&	22.34  & 10.1&&	-0.71&	29.61 & 11.1\\
CM&&    -2.21&	25.37& 10.2   &&	3.97&	42.58& 12.8    &&	1.61&	22.46& 10.1&&	1.55&	30.49& 11.5\\
CS&&     1.09&	24.50& 9.8    &&	-2.33&	30.48& 9.1     &&	3.36&	25.75& 11.6&&	-1.50&	24.27& 9.1\\
E&&      -1.87&	23.44& 9.4   &&	-0.16&	37.61& 11.2    &&	0.22&	25.71& 11.6&&	-0.80&	28.27& 10.6\\
F&&      8.61&	40.32& 16.2    &&	-1.58&	19.17& 5.8 &&	-8.63&	14.51& 6.5 &&	0.54&	22.15& 8.3\\
HC&&    -1.81&	20.04&  8      &&	-0.17&	35.02& 10.5&&	-0.53&	14.94& 6.7&&	-0.56&	26.45& 9.9\\
IN&&      3.41&	33.15& 13.3&&	-0.43&	39.94& 12.0    &&	-0.77&	24.64& 11.1&&	0.17&	31.47& 11.9\\
IT&&    -1.46&	25.93& 10.5&&	2.84&	44.45& 13.3     &&	3.30&	24.71& 11.5&&	2.23&	33.71& 12.7\\
M&&     -0.79&	15.93& 6.2 &&	-0.99&	20.01& 5.9     &&	-3.06&	19.40& 8.7 &&	-0.98&	15.84& 6.0\\
RE&&    -0.89&	14.47&	5.8&&	-0.55 & 14.68& 4.4     &&	-1.48&	12.14& 5.4&&	0.51&	14.22& 5.3\\
U&&     -1.09&	10.27&	4.1&&	-0.66 & 9.70& 2.9      &&	2.33&	14.97& 6.7 &&	-0.43&	9.59& 3.6\\
\bottomrule
 \end{tabular}
  \caption*{\scriptsize \textit{Notes}: The table shows the average \textsc{net} and \textsc{agg} network characteristics with respect to the 11 U.S. industries. When the \textsc{net} measure is positive the \indvix can be classified as a \textsc{net} marginal transmitter, while, when negative, it can be classified as a \textsc{net}  marginal receiver. The \textsc{agg} statistic is computed as the sum (in absolute values) between \textsc{to} and \textsc{from}. The highest values of the \textsc{agg} statistic are associated with uncertainty hubs, while the lowest with uncertainty non-hubs. The statistics are reported for each business cycle in our sample, namely inversion, recession, expansion separately and also for the total period, namely from 03-01-2000 to 29-05-2020, at a daily frequency.
  }
  }
  \end{threeparttable}
\end{table}

\clearpage
\newpage
\section{Robustness Checks}

\setcounter{table}{0}
\setcounter{table}{0}
\renewcommand{\thetable}{F\arabic{table}}

\renewcommand{\thefigure}{F\arabic{figure}}

\setcounter{figure}{0}

\begin{table}[h!]
 \centering
 \caption{ADS Predictive Results} \label{ADSPre}
 \begin{threeparttable}
 \centering
\scriptsize{
\begin{tabular}{clcccccc}
  \toprule
&&  \multicolumn{5}{c}{Panel A: ADS}\\
\toprule
&&  \multicolumn{1}{c}{\textit{h=1}} &  \multicolumn{1}{c}{\textit{h=3}} & \multicolumn{1}{c}{\textit{h=6}} & \multicolumn{1}{c}{\textit{h=9}} & \multicolumn{1}{c}{\textit{h=12}}\\
  \cline{1-3} \cline{4-4} \cline{5-5} \cline{6-6} \cline{7-7}
\toprule
$\mathcal{C}_t$ $\mid$ $X_t$  & &  -0.020 & -0.032* & -0.036** & -0.055*** & -0.038** \\
  & & (0.016) & (0.019) & (0.019) & (0.017) & (0.019) \\
Adj. $R^{2}$ & & \multicolumn{1}{c}{0.273} & \multicolumn{1}{c}{0.030} & \multicolumn{1}{c}{0.029} & \multicolumn{1}{c}{0.063} & \multicolumn{1}{c}{0.072} \\
Obs &  & \multicolumn{1}{c}{243} & \multicolumn{1}{c}{241} & \multicolumn{1}{c}{238} & \multicolumn{1}{c}{235} & \multicolumn{1}{c}{232} \\
\toprule
&&  \multicolumn{5}{c}{Panel B: ADS Expansion}\\
\toprule
&&  \multicolumn{1}{c}{\textit{h=1}} &  \multicolumn{1}{c}{\textit{h=3}} & \multicolumn{1}{c}{\textit{h=6}} & \multicolumn{1}{c}{\textit{h=9}} & \multicolumn{1}{c}{\textit{h=12}}\\
  \cline{1-3} \cline{4-4} \cline{5-5} \cline{6-6} \cline{7-7}
\toprule
$\mathcal{C}_t$ $\mid$ $X_t$  & &   -0.002 & -0.003 & -0.005* & -0.007*** & -0.007*** \\
  & & (0.003) & (0.003) & (0.003) & (0.003) & (0.003) \\
Adj. $R^{2}$ & & \multicolumn{1}{c}{0.009} & \multicolumn{1}{c}{0.055} & \multicolumn{1}{c}{0.065} & \multicolumn{1}{c}{0.082} & \multicolumn{1}{c}{0.056} \\
Obs &  & \multicolumn{1}{c}{243} & \multicolumn{1}{c}{241} & \multicolumn{1}{c}{238} & \multicolumn{1}{c}{235} & \multicolumn{1}{c}{232} \\
\toprule
&&  \multicolumn{5}{c}{Panel C: ADS Recession}\\
\toprule
&&  \multicolumn{1}{c}{\textit{h=1}} &  \multicolumn{1}{c}{\textit{h=3}} & \multicolumn{1}{c}{\textit{h=6}} & \multicolumn{1}{c}{\textit{h=9}} & \multicolumn{1}{c}{\textit{h=12}}\\
\toprule
$\mathcal{C}_t$ $\mid$ $X_t$  & &   -0.022 & -0.025 & -0.031* & -0.048*** & -0.031* \\
  & & (0.016) & (0.019) & (0.019) & (0.019) & (0.019) \\
Adj. $R^{2}$ & & \multicolumn{1}{c}{0.273} & \multicolumn{1}{c}{0.022} & \multicolumn{1}{c}{0.025} & \multicolumn{1}{c}{0.053} & \multicolumn{1}{c}{0.067} \\
Obs &  & \multicolumn{1}{c}{243} & \multicolumn{1}{c}{241} & \multicolumn{1}{c}{238} & \multicolumn{1}{c}{235} & \multicolumn{1}{c}{232} \\
\bottomrule
\end{tabular}
\begin{tablenotes}
\item {\scriptsize \textit{Notes}: This table presents the results of the predictive regression in equation \ref{PredictregresstotControl} between the aggregate network connectedness and the 3-month moving average of the ADS indicator of business cycle (Panel A). In Panel B and Panel C the results of the predictive regression with respect to the ADS expansion and recession indicators are reported, respectively. We also add a set of controls, $X$. The five columns of the table represent different predictability horizons with $h \in (1,3,6,9,12)$. Regressions' coefficients and standard errors (in parentheses), and adjusted-$R^2$ are reported. Coefficients are marked with *, **, *** for 10\%, 5\%, 1\% significance levels, respectively. Intercept and controls results are not reported for the sake of space. Series are considered at a monthly frequency between 01-2000 and 05-2020.}
\end{tablenotes}
}
\end{threeparttable}
\end{table}

\begin{table}[h!]
 \centering
 \caption{ADS Predictive Results} \label{IPU.S.CIPre}
 \begin{threeparttable}
 \centering
\scriptsize{
\begin{tabular}{clcccccc}
  \toprule
&&  \multicolumn{5}{c}{Panel A: IP Growth Rate}\\
\toprule
&&  \multicolumn{1}{c}{\textit{h=1}} &  \multicolumn{1}{c}{\textit{h=3}} & \multicolumn{1}{c}{\textit{h=6}} & \multicolumn{1}{c}{\textit{h=9}} & \multicolumn{1}{c}{\textit{h=12}}\\
\toprule
$\mathcal{C}_t$ $\mid$ $X_t$  & & -0.001 & -0.003** & -0.003** & -0.003*** & -0.003** \\
 & & (0.001) & (0.001) & (0.001) & (0.001) & (0.001) \\
Adj. $R^{2}$ & &  \multicolumn{1}{c}{0.247} & \multicolumn{1}{c}{0.054} & \multicolumn{1}{c}{0.035} & \multicolumn{1}{c}{0.057} & \multicolumn{1}{c}{0.062} \\
Obs &  & \multicolumn{1}{c}{243} & \multicolumn{1}{c}{241} & \multicolumn{1}{c}{238} & \multicolumn{1}{c}{235} & \multicolumn{1}{c}{232} \\
\toprule
&&  \multicolumn{5}{c}{Panel B: U.S.CI Growth Rate}\\
\toprule
&&  \multicolumn{1}{c}{\textit{h=1}} &  \multicolumn{1}{c}{\textit{h=3}} & \multicolumn{1}{c}{\textit{h=6}} & \multicolumn{1}{c}{\textit{h=9}} & \multicolumn{1}{c}{\textit{h=12}}\\
\toprule
$\mathcal{C}_t$ $\mid$ $X_t$  & &   -0.019 & -0.060** & -0.087*** & -0.150*** & -0.140*** \\
 & & (0.027) & (0.029) & (0.030) & (0.027) & (0.026) \\
Adj. $R^{2}$ & & \multicolumn{1}{c}{0.009} & \multicolumn{1}{c}{0.055} & \multicolumn{1}{c}{0.065} & \multicolumn{1}{c}{0.082} & \multicolumn{1}{c}{0.056} \\
Obs &  & \multicolumn{1}{c}{0.205} & \multicolumn{1}{c}{0.078} & \multicolumn{1}{c}{0.058} & \multicolumn{1}{c}{0.213} & \multicolumn{1}{c}{0.280} \\
\bottomrule
\end{tabular}
\begin{tablenotes}
\item {\scriptsize \textit{Notes}: This table presents the results of the predictive regression in equation \ref{PredictregresstotControl} between the aggregate network connectedness, and the industrial production growth rate in Panel A, and the U.S. coincident indicator in Panel B. We also add a set of controls, $X$. The five columns of the table represent different predictability horizons with $h \in (1,3,6,9,12)$. Regressions' coefficients and standard errors (in parentheses), and adjusted-$R^2$ are reported. Coefficients are marked with *, **, *** for 10\%, 5\%, 1\% significance levels, respectively. Intercept and controls results are not reported for the sake of space. Series are considered at a monthly frequency between 01-2000 and 05-2020.}
\end{tablenotes}
}
\end{threeparttable}
\end{table}

\begin{table}[h!]
 \centering
 \caption{Coincident Indicators Predictive Results (2) } \label{CFNAIPreU.S.CIcontrol}
 \begin{threeparttable}
 \centering
\scriptsize{
\begin{tabular}{clcccccc}
\toprule
&&  \multicolumn{5}{c}{Dependent: CFNAI-3M}\\
\toprule
&&  \multicolumn{1}{c}{\textit{h=1}} &  \multicolumn{1}{c}{\textit{h=3}} & \multicolumn{1}{c}{\textit{h=6}} & \multicolumn{1}{c}{\textit{h=9}} & \multicolumn{1}{c}{\textit{h=12}}\\
\toprule
$\mathcal{C}_t$ $\mid$ $X_t$  & &  -0.008** & -0.021*** & -0.023*** & -0.039*** & -0.028*** \\
  & & (0.004) & (0.007) & (0.007) & (0.007) & (0.007) \\
U.S.CI & & 0.104*** & 0.064*** & 0.044*** & 0.025** & 0.003 \\
  & & (0.006) & (0.010) & (0.011) & (0.011) & (0.011) \\
Adj. $R^{2}$ & & \multicolumn{1}{c}{0.678} & \multicolumn{1}{c}{0.280} & \multicolumn{1}{c}{0.123} & \multicolumn{1}{c}{0.191} & \multicolumn{1}{c}{0.164} \\
Obs &  & \multicolumn{1}{c}{243} & \multicolumn{1}{c}{241} & \multicolumn{1}{c}{238} & \multicolumn{1}{c}{235} & \multicolumn{1}{c}{232} \\
\toprule
&&  \multicolumn{5}{c}{Dependent: ADS}\\
\toprule
&&  \multicolumn{1}{c}{\textit{h=1}} &  \multicolumn{1}{c}{\textit{h=3}} & \multicolumn{1}{c}{\textit{h=6}} & \multicolumn{1}{c}{\textit{h=9}} & \multicolumn{1}{c}{\textit{h=12}}\\
\toprule
$\mathcal{C}_t$ $\mid$ $X_t$  & &   -0.012 & -0.023 & -0.033** & -0.052*** & -0.037** \\
  & & (0.014) & (0.018) & (0.019) & (0.019) & (0.019) \\
U.S.CI & & 0.172*** & 0.088*** & 0.059** & 0.041 & 0.023 \\
  & & (0.020) & (0.029) & (0.029) & (0.029) & (0.029) \\
Adj. $R^{2}$ & & \multicolumn{1}{c}{0.440} & \multicolumn{1}{c}{0.064} & \multicolumn{1}{c}{0.042} & \multicolumn{1}{c}{0.068} & \multicolumn{1}{c}{0.071} \\
Obs &  & \multicolumn{1}{c}{243} & \multicolumn{1}{c}{241} & \multicolumn{1}{c}{238} & \multicolumn{1}{c}{235} & \multicolumn{1}{c}{232} \\
\bottomrule
\end{tabular}
\begin{tablenotes}
\item {\scriptsize \textit{Notes}: This table presents the results of the predictive regressions between the aggregate network connectedness, and the coincident indicators of business cycle, namely CFNAI and ADS. We present results for regression equation \ref{PredictregresstotControl} in which we add a set of controls including also the leading indicator, U.S.LI. The five columns of the table represent different predictability horizons with $h \in (1,3,6,9,12)$. Regressions' coefficients and standard errors (in parentheses), and adjusted-$R^2$ are reported. Coefficients are marked with *, **, *** for 10\%, 5\%, 1\% significance levels, respectively. Intercept and controls results are not reported for the sake of space, exception for the U.S.LI control. Series are considered at a monthly frequency between 01-2000 and 05-2020. }
\end{tablenotes}
}
\end{threeparttable}
\end{table}

\clearpage
\newpage
\section{Hubs and non-hubs predictive results }

\setcounter{table}{0}
\setcounter{table}{0}
\renewcommand{\thetable}{G\arabic{table}}

\renewcommand{\thefigure}{G\arabic{figure}}

\setcounter{figure}{0}

\begin{table}[h!]
 \centering
 \caption{CFNAI-MA3 Hubs Network Predictive Results} \label{CFNAI-MA3Hubs}
 \begin{threeparttable}
 \centering
\scriptsize{
\begin{tabular}{clcccccc}
\toprule
&&  \multicolumn{5}{c}{Panel A: CFNAI-MA3}\\
\toprule
  &&  \multicolumn{1}{c}{\textit{h=1}} &  \multicolumn{1}{c}{\textit{h=3}} & \multicolumn{1}{c}{\textit{h=6}} & \multicolumn{1}{c}{\textit{h=9}} & \multicolumn{1}{c}{\textit{h=12}}\\
\toprule
 $\mathcal{C}_t^{\text{hub}}$ $\mid$ $X_t$ & & -0.009** & -0.016*** & -0.019*** & -0.023*** & -0.018*** \\
  & & (0.004) & (0.004) & (0.004) & (0.004) & (0.004) \\
 Adj. $R^{2}$ &  & \multicolumn{1}{c}{0.302} & \multicolumn{1}{c}{0.176} & \multicolumn{1}{c}{0.104} & \multicolumn{1}{c}{0.178} & \multicolumn{1}{c}{0.186} \\
 Obs & & \multicolumn{1}{c}{243} & \multicolumn{1}{c}{241} & \multicolumn{1}{c}{238} & \multicolumn{1}{c}{235} & \multicolumn{1}{c}{232} \\
\toprule
&&  \multicolumn{5}{c}{Panel B: CFNAI-MA3 Expansion}\\
\toprule
  &&  \multicolumn{1}{c}{\textit{h=1}} &  \multicolumn{1}{c}{\textit{h=3}} & \multicolumn{1}{c}{\textit{h=6}} & \multicolumn{1}{c}{\textit{h=9}} & \multicolumn{1}{c}{\textit{h=12}}\\
\toprule
$\mathcal{C}_t^{\text{hub}}$ $\mid$ $X_t$ & & -0.004*** & -0.005*** & -0.006*** & -0.006*** & -0.006*** \\
  && (0.001) & (0.001) & (0.001) & (0.001) & (0.001) \\
 Adj. $R^{2}$ &  & \multicolumn{1}{c}{0.099} & \multicolumn{1}{c}{0.152} & \multicolumn{1}{c}{0.238} & \multicolumn{1}{c}{0.290} & \multicolumn{1}{c}{0.289} \\
 Obs & & \multicolumn{1}{c}{243} & \multicolumn{1}{c}{241} & \multicolumn{1}{c}{238} & \multicolumn{1}{c}{235} & \multicolumn{1}{c}{232} \\
 \toprule
&&  \multicolumn{5}{c}{Panel C: CFNAI-MA3 Recession}\\
\toprule
  &&  \multicolumn{1}{c}{\textit{h=1}} &  \multicolumn{1}{c}{\textit{h=3}} & \multicolumn{1}{c}{\textit{h=6}} & \multicolumn{1}{c}{\textit{h=9}} & \multicolumn{1}{c}{\textit{h=12}}\\
\toprule
 $\mathcal{C}_t^{\text{hub}}$ $\mid$ $X_t$ & & -0.005 & -0.012*** & -0.013*** & -0.017*** & -0.012*** \\
  && (0.003) & (0.004) & (0.004) & (0.004) & (0.004) \\
 Adj. $R^{2}$ &  & \multicolumn{1}{c}{0.291} & \multicolumn{1}{c}{0.153} & \multicolumn{1}{c}{0.087} & \multicolumn{1}{c}{0.148} & \multicolumn{1}{c}{0.153} \\
 Obs & & \multicolumn{1}{c}{243} & \multicolumn{1}{c}{241} & \multicolumn{1}{c}{238} & \multicolumn{1}{c}{235} & \multicolumn{1}{c}{232} \\
\bottomrule
\end{tabular}
\begin{tablenotes}
\item {\scriptsize \textit{Notes}: This table presents the results of the predictive regression \ref{Predicthubsmulti} (with $\beta_{non-hub}$ set equal to 0) considering only the predicting results of the uncertainty hub sub-network with respect to the 3-month moving average of the Chicago FED National Activity Index (CFNAI-MA3), indicator of business cycle (Panel A). In Panel B and Panel C the results of the predictive regression with respect to the CFNAI-MA3 expansion and recession indicators are reported, respectively. The five columns of the table represent different predictability horizons with $h \in (1,3,6,9,12)$. Regressions' coefficients and standard errors (in parentheses), and adjusted-$R^2$ are reported. Coefficients are marked with *, **, *** for 10\%, 5\%, 1\% significance levels, respectively. Intercept and controls results are not reported for the sake of space. Series are considered at a monthly frequency between 01-2000 and 05-2020. }
\end{tablenotes}
}
\end{threeparttable}
\end{table}

\begin{table}[h!]
 \centering
 \caption{Stricter Hubs vs non-Hubs Network Predictive Results} \label{CFNAIHubsnoHubsSTRICT}
 \begin{threeparttable}
 \centering
\scriptsize{
\begin{tabular}{clcccccc}
 \toprule
&&  \multicolumn{5}{c}{Panel A: CFNAI-MA3}\\
\toprule
   &&  \multicolumn{1}{c}{\textit{h=1}} &  \multicolumn{1}{c}{\textit{h=3}} & \multicolumn{1}{c}{\textit{h=6}} & \multicolumn{1}{c}{\textit{h=9}} & \multicolumn{1}{c}{\textit{h=12}}\\
 \toprule
 $\mathcal{C}_t^{\text{hub}}$ $\mid$ $X_t$ & & -0.010*** & -0.013*** & -0.016*** & -0.017*** & -0.016*** \\
  &&  (0.003) & (0.004) & (0.004) & (0.004) & (0.004) \\
    $\mathcal{C}_t^{\text{non-hub}}$ $\mid$ $X_t$&& -0.014** & -0.022*** & -0.015** & -0.013* & -0.0003 \\
  &&  (0.006) & (0.006) & (0.007) & (0.007) & (0.007) \\
 Adj. $R^{2}$ &  &   \multicolumn{1}{c}{0.334} & \multicolumn{1}{c}{0.224} & \multicolumn{1}{c}{0.121} & \multicolumn{1}{c}{0.163} & \multicolumn{1}{c}{0.173} \\
 Obs & & \multicolumn{1}{c}{243} & \multicolumn{1}{c}{241} & \multicolumn{1}{c}{238} & \multicolumn{1}{c}{235} & \multicolumn{1}{c}{232} \\
  \toprule
&&  \multicolumn{5}{c}{Panel B: CFNAI-MA3 Expansion}\\
\toprule
   &&  \multicolumn{1}{c}{\textit{h=1}} &  \multicolumn{1}{c}{\textit{h=3}} & \multicolumn{1}{c}{\textit{h=6}} & \multicolumn{1}{c}{\textit{h=9}} & \multicolumn{1}{c}{\textit{h=12}}\\
 \toprule
 $\mathcal{C}_t^{\text{hub}}$ $\mid$ $X_t$ & & -0.003** & -0.004*** & -0.004*** & -0.005*** & -0.006*** \\
  && (0.001) & (0.001) & (0.001) & (0.001) & (0.001) \\
  $\mathcal{C}_t^{\text{non-hub}}$ $\mid$ $X_t$&& 0.0004 & -0.002 & -0.005*** & -0.002 & -0.0002 \\
  && (0.002) & (0.002) & (0.002) & (0.002) & (0.002) \\
 Adj. $R^{2}$ &  &  \multicolumn{1}{c}{0.073} & \multicolumn{1}{c}{0.135} & \multicolumn{1}{c}{0.240} & \multicolumn{1}{c}{0.297} & \multicolumn{1}{c}{0.293} \\
 Obs & & \multicolumn{1}{c}{243} & \multicolumn{1}{c}{241} & \multicolumn{1}{c}{238} & \multicolumn{1}{c}{235} & \multicolumn{1}{c}{232} \\
 \toprule
&&  \multicolumn{5}{c}{Panel C: CFNAI-MA3 Recession}\\
\toprule
   &&  \multicolumn{1}{c}{\textit{h=1}} &  \multicolumn{1}{c}{\textit{h=3}} & \multicolumn{1}{c}{\textit{h=6}} & \multicolumn{1}{c}{\textit{h=9}} & \multicolumn{1}{c}{\textit{h=12}}\\
 \toprule
 $\mathcal{C}_t^{\text{hub}}$ $\mid$ $X_t$ & & -0.007** & -0.010*** & -0.012*** & -0.012*** & -0.010*** \\
  && (0.003) & (0.004) & (0.004) & (0.004) & (0.004) \\
   $\mathcal{C}_t^{\text{non-hub}}$ $\mid$ $X_t$&& -0.015*** & -0.020*** & -0.009 & -0.012* & -0.0002 \\
  && (0.005) & (0.006) & (0.006) & (0.006) & (0.006) \\
 Adj. $R^{2}$ &  &   \multicolumn{1}{c}{0.326} & \multicolumn{1}{c}{0.200} & \multicolumn{1}{c}{0.098} & \multicolumn{1}{c}{0.135} & \multicolumn{1}{c}{0.141} \\
 Obs & & \multicolumn{1}{c}{243} & \multicolumn{1}{c}{241} & \multicolumn{1}{c}{238} & \multicolumn{1}{c}{235} & \multicolumn{1}{c}{232} \\
 \toprule
&&  \multicolumn{5}{c}{Panel D: CLI}\\
\toprule
   &&  \multicolumn{1}{c}{\textit{h=1}} &  \multicolumn{1}{c}{\textit{h=3}} & \multicolumn{1}{c}{\textit{h=6}} & \multicolumn{1}{c}{\textit{h=9}} & \multicolumn{1}{c}{\textit{h=12}}\\
 \toprule
 $\mathcal{C}_t^{\text{hub}}$ $\mid$ $X_t$ & & -0.013** & -0.024*** & -0.038*** & -0.048*** & -0.050*** \\
  & & (0.005) & (0.006) & (0.006) & (0.006) & (0.006) \\
  $\mathcal{C}_t^{\text{non-hub}}$ $\mid$ $X_t$ && -0.010 & -0.017* & -0.019* & -0.022** & -0.011 \\
  &&  (0.008) & (0.009) & (0.010) & (0.010) & (0.010) \\
 Adj. $R^{2}$ &  &  \multicolumn{1}{c}{0.451} & \multicolumn{1}{c}{0.341} & \multicolumn{1}{c}{0.285} & \multicolumn{1}{c}{0.338} & \multicolumn{1}{c}{0.334} \\
 Obs & & \multicolumn{1}{c}{243} & \multicolumn{1}{c}{241} & \multicolumn{1}{c}{238} & \multicolumn{1}{c}{235} & \multicolumn{1}{c}{232} \\
 \toprule
&&  \multicolumn{5}{c}{Panel E: CFNAI-MA3 controlling for CLI}\\
\toprule
   &&  \multicolumn{1}{c}{\textit{h=1}} &  \multicolumn{1}{c}{\textit{h=3}} & \multicolumn{1}{c}{\textit{h=6}} & \multicolumn{1}{c}{\textit{h=9}} & \multicolumn{1}{c}{\textit{h=12}}\\
 \toprule
 $\mathcal{C}_t^{\text{hub}}$ $\mid$ $X_t$ & & -0.006** & -0.010*** & -0.014*** & -0.016*** & -0.016*** \\
  && (0.003) & (0.003) & (0.004) & (0.004) & (0.004) \\
   $\mathcal{C}_t^{\text{non-hub}}$ $\mid$ $X_t$&& -0.013*** & -0.022*** & -0.015** & -0.013* & -0.0003 \\
  && (0.004) & (0.005) & (0.006) & (0.006) & (0.007) \\
 CLI &&   0.418*** & 0.336*** & 0.241*** & 0.112** & -0.017 \\
  && (0.036) & (0.044) & (0.051) & (0.052) & (0.053) \\
 Adj. $R^{2}$ &  &  \multicolumn{1}{c}{0.579} & \multicolumn{1}{c}{0.378} & \multicolumn{1}{c}{0.197} & \multicolumn{1}{c}{0.177} & \multicolumn{1}{c}{0.170} \\
 Obs & & \multicolumn{1}{c}{243} & \multicolumn{1}{c}{241} & \multicolumn{1}{c}{238} & \multicolumn{1}{c}{235} & \multicolumn{1}{c}{232} \\
 \bottomrule
\end{tabular}
\begin{tablenotes}
\item {\scriptsize \textit{Notes}: This table presents the results of the predictive regression \ref{Predicthubsmulti} comparing the predictive ability of the uncertainty hubs (CM, IN and IT) vs non-hubs (M, RE and U) sub-networks, with respect to the 3-month moving average of the Chicago FED National Activity Index (CFNAI-MA3) in Panel A. In Panel B and Panel C the results of the predictive regression with respect to the CFNAI expansion and recession periods are reported, respectively. In Panel D the results with respect to the leading indicator CLI are reported. In Panel E the results of predicting the coincident business cycle indicator controlling for CLI are reported.  The five columns of the table represent different predictability horizons with $h \in (1,3,6,9,12)$. Regressions' coefficients and standard errors (in parentheses), and adjusted-$R^2$ are reported. Coefficients are marked with *, **, *** for 10\%, 5\%, 1\% significance levels, respectively. Intercept and controls results are not reported for the sake of space, the only exception being the CLI control. Series are considered at a monthly frequency between 01-2000 and 05-2020. }
\end{tablenotes}
}
\end{threeparttable}
\end{table}

\clearpage
\newpage
\section{U.S. GDP and its volatility}

\setcounter{table}{0}
\setcounter{table}{0}
\renewcommand{\thetable}{H\arabic{table}}

\renewcommand{\thefigure}{H\arabic{figure}}

\setcounter{figure}{0}

\begin{table}[h!]
 \centering
 \caption{U.S. GDP Predictive Results} \label{GDPGRPre}
 \begin{threeparttable}
 \centering
\scriptsize{
\begin{tabular}{clccccc}
 \toprule
&&  \multicolumn{4}{c}{Panel A: GDP Growth Rate}\\
\toprule
  &&  \multicolumn{1}{c}{\textit{h=1}} &  \multicolumn{1}{c}{\textit{h=2}} & \multicolumn{1}{c}{\textit{h=3}} & \multicolumn{1}{c}{\textit{h=4}}\\
 \toprule
 $\mathcal{C}_t$ $\mid$ $X_t$    & & -0.056 & -0.158*** & -0.176*** & -0.156*** \\
  & &(0.046) & (0.041) & (0.041) & (0.041) \\
Adj. $R^{2}$ &   &    \multicolumn{1}{c}{0.015} & \multicolumn{1}{c}{0.221} & \multicolumn{1}{c}{0.218} & \multicolumn{1}{c}{0.250} \\
Obs & & \multicolumn{1}{c}{79} & \multicolumn{1}{c}{78} & \multicolumn{1}{c}{77} & \multicolumn{1}{c}{76} \\
 \toprule
&&  \multicolumn{4}{c}{Panel B: GDP Volatility}\\
\toprule
  &&  \multicolumn{1}{c}{\textit{h=1}} &  \multicolumn{1}{c}{\textit{h=2}} & \multicolumn{1}{c}{\textit{h=3}} & \multicolumn{1}{c}{\textit{h=4}}\\
\toprule
 $\mathcal{C}_t$ $\mid$ $X_t$ && 0.028* & 0.036** & 0.047*** & 0.039*** \\
  & & (0.014) & (0.014) & (0.014) & (0.013) \\
Adj. $R^{2}$ &   & \multicolumn{1}{c}{0.138} & \multicolumn{1}{c}{0.106} & \multicolumn{1}{c}{0.165} & \multicolumn{1}{c}{0.219} \\
Obs & & \multicolumn{1}{c}{79} & \multicolumn{1}{c}{78} & \multicolumn{1}{c}{77} & \multicolumn{1}{c}{76} \\
\bottomrule
\end{tabular}
\begin{tablenotes}
\item {\scriptsize \textit{Notes}: This table presents the results of the predictive regression \ref{PredictregresstotControl} between the aggregate network connectedness, and the U.S. GDP growth rate (Panel A) and GDP volatility (Panel B). The four columns of the table represent different predictability horizons with $h \in (1,2,3,4)$ quarters.   Regressions' coefficients and standard errors (in parentheses), and adjusted-$R^2$ are reported. Coefficients are marked with *, **, *** for 10\%, 5\%, 1\% significance levels, respectively. Intercept and controls results are not reported for the sake of space. Series are all taken at quarterly frequency, between 01-2000 and 05-2020.  }
\end{tablenotes}
}
\end{threeparttable}
\end{table}

\begin{table}[h!]
 \centering
 \caption{U.S. GDP Prediction with Hubs and non-Hubs} \label{GDPGRPreHUBS}
 \begin{threeparttable}
 \centering
\scriptsize{
\begin{tabular}{clccccc}
 \toprule
&&  \multicolumn{4}{c}{Panel A: GDP Growth Rate}\\
\toprule
  &&  \multicolumn{1}{c}{\textit{h=1}} &  \multicolumn{1}{c}{\textit{h=2}} & \multicolumn{1}{c}{\textit{h=3}} & \multicolumn{1}{c}{\textit{h=4}}\\
 \toprule
 $\mathcal{C}_t^{\text{hub}}$ &   & -0.055** & -0.083*** & -0.104*** & -0.081*** \\
  & (0.027) & (0.024) & (0.024) & (0.024) \\
  $\mathcal{C}_t^{\text{non-hub}}$ & & 0.017 & -0.043 & -0.017 & -0.039 \\
  & (0.044) & (0.039) & (0.039) & (0.039) \\
Adj. $R^{2}$ &   &   \multicolumn{1}{c}{0.038} & \multicolumn{1}{c}{0.223} & \multicolumn{1}{c}{0.238} & \multicolumn{1}{c}{0.246} \\
Obs & & \multicolumn{1}{c}{79} & \multicolumn{1}{c}{78} & \multicolumn{1}{c}{77} & \multicolumn{1}{c}{76} \\
 \toprule
&&  \multicolumn{4}{c}{Panel B: GDP Volatility}\\
\toprule
  &&  \multicolumn{1}{c}{\textit{h=1}} &  \multicolumn{1}{c}{\textit{h=2}} & \multicolumn{1}{c}{\textit{h=3}} & \multicolumn{1}{c}{\textit{h=4}}\\
\toprule
$\mathcal{C}_t^{\text{hub}}$ &   & 0.030*** & 0.029*** & 0.024*** & 0.015* \\
  && (0.008) & (0.008) & (0.008) & (0.008) \\
 $\mathcal{C}_t^{\text{non-hub}}$ &    & -0.016 & -0.004 & 0.014 & 0.024* \\
  & &  (0.013) & (0.013) & (0.013) & (0.013) \\
Adj. $R^{2}$ &   &    \multicolumn{1}{c}{0.236} & \multicolumn{1}{c}{0.166} & \multicolumn{1}{c}{0.170} & \multicolumn{1}{c}{0.221} \\
Obs & & \multicolumn{1}{c}{79} & \multicolumn{1}{c}{78} & \multicolumn{1}{c}{77} & \multicolumn{1}{c}{76} \\
\bottomrule
\end{tabular}
\begin{tablenotes}
\item {\scriptsize \textit{Notes}: This table presents the results of the predictive regression \ref{PredictregresstotControl} between the hubs and non-hubs networks, and the U.S. GDP growth rate (Panel A) and GDP volatility (Panel B). The four columns of the table represent different predictability horizons with $h \in (1,2,3,4)$ quarters.  Regressions' coefficients and standard errors (in parentheses), and adjusted-$R^2$ are reported. Coefficients are marked with *, **, *** for 10\%, 5\%, 1\% significance levels, respectively. Intercept and controls results are not reported for the sake of space. Series are all taken at quarterly frequency, between 01-2000 and 05-2020.  }
\end{tablenotes}
}
\end{threeparttable}
\end{table}

\end{appendices}

\end{document}